\documentclass[letterpaper]{article} 
\usepackage{aaai2026}  
\usepackage{times}  
\usepackage{helvet}  
\usepackage[dvipsnames]{xcolor}
\usepackage{makecell}
\usepackage{pifont}
\usepackage{subcaption}
\usepackage{courier}  
\usepackage[hyphens]{url}  
\usepackage{graphicx} 
\usepackage{natbib}  
\usepackage{caption} 
\usepackage{algorithm}
\usepackage{algorithmic}

\usepackage{newfloat}
\usepackage{listings}
\DeclareCaptionStyle{ruled}{labelfont=normalfont,labelsep=colon,strut=off} 
\floatstyle{ruled}
\newfloat{listing}{tb}{lst}{}
\floatname{listing}{Listing}
\usepackage{amsthm}
\usepackage{amsmath,amssymb,amsfonts}
\usepackage{bm}
\usepackage{empheq}
\usepackage{multirow}
\usepackage{array,makecell,booktabs}

\usepackage{fancyvrb}
\newcommand{\cmark}{\textcolor{ForestGreen}{\ding{51}}}%
\newcommand{\xmark}{\textcolor{Maroon}{\ding{55}}}%
\newcommand{\namark}{\textcolor{Maroon}{\textbf{?}}}%

\usepackage[utf8]{inputenc} 
\usepackage[T1]{fontenc}    
\usepackage{url}            
\usepackage{booktabs}       
\usepackage{amsfonts}       
\usepackage{nicefrac}       
\usepackage{microtype}      
\usepackage{xcolor}         

\usepackage{multirow}   
\usepackage{booktabs}    
\usepackage{amsthm}
\usepackage{graphicx}
\usepackage{amssymb}
\usepackage{amsmath}  
\usepackage{bm}
\usepackage{mathrsfs}
\usepackage{graphicx}
\usepackage{color}
\usepackage{svg}
\usepackage{subcaption}
\usepackage{threeparttable}

\usepackage{enumitem}

\usepackage{algorithm}
\usepackage{algorithmic}
\usepackage{xcolor}

\usepackage{natbib}

\newtheorem{remark}{Remark}
\newtheorem{definition}{Definition}
\newtheorem{theorem}{Theorem}
\newtheorem{proposition}{Proposition}
\newtheorem{lemma}{Lemma}

\newtheorem{assumption}{Assumption}

\newcommand{\Int}{\operatorname{Int}}

\newcommand{\Inthull}{\operatorname{IntHull}}
\DeclareMathOperator*{\rank}{rank}

\DeclareMathOperator*{\rows}{rows}

\DeclareMathOperator*{\relind}{relind}
\DeclareMathOperator*{\aff}{aff}

\title{Efficient Verification and Falsification of ReLU Neural Barrier Certificates}

\author{
    Dejin Ren\textsuperscript{\rm 1, 2}\equalcontrib, Yiling Xue\textsuperscript{\rm 1, 2, 3}\equalcontrib, Taoran Wu\textsuperscript{\rm 1, 2}, and Bai Xue\textsuperscript{\rm 1, 2, 3}\thanks{Corresponding author}
}
\affiliations{

\textsuperscript{\rm 1} Key laboratory of System Software (Chinese Academy of Sciences) and State Key
Laboratory of Computer Sciences, Institute of Software, Chinese Academy of
Sciences, Beijing, China\\
\textsuperscript{\rm 2} University of Chinese Academy of Sciences, Beijing, China\\
\textsuperscript{\rm 3} School of Advanced Interdisciplinary Sciences, University of Chinese Academy of Sciences, Beijing, China\\
\{rendj,xueyl,wutr,xuebai\}@ios.ac.cn
}

\usepackage{bibentry}

\begin{document}

\maketitle

\begin{abstract}
Barrier certificates play an important role in verifying the safety of continuous-time systems, including autonomous driving, robotic manipulators and other critical applications. Recently, 
ReLU neural barrier certificates---barrier certificates represented by the ReLU neural networks---have attracted significant attention in the safe control community due to their promising performance.
However, because of the approximate nature of neural networks, rigorous verification methods are required to ensure the correctness of these certificates. This paper presents  a necessary and sufficient condition for verifying  the correctness of ReLU
neural barrier certificates. The proposed condition can be encoded  as either an Satisfiability Modulo Theories (SMT) or  optimization problem, enabling both verification and falsification.  To the best of our knowledge, this is the first approach
capable of falsifying ReLU neural barrier certificates. Numerical experiments demonstrate the validity and 
effectiveness of the proposed method in both verifying and falsifying such certificates.
\end{abstract}

\begin{links}
    \link{Code}{https://github.com/YilingXue/evf-rnbc}
\end{links}

\section{Introduction}
\label{sec:intro}



Safety is a crucial property for continuous-time systems, including autonomous driving, robotic manipulators and other vital applications.
Formally, a system is safe if every trajectory starting from the initial set never enters the unsafe set. In practice, a \textit{barrier certificate} offers a theoretical guarantee of safety. The 0-superlevel (or sublevel) set of a barrier certificate defines a positive invariant set --- meaning that trajectories starting within it remain there indefinitely. If this positive invariant set contains the initial set and does not intersect the unsafe set, the safety of the system is ensured.


In recent years, the sum-of-squares (SOS) technique has been widely used to synthesize polynomial barrier certificates for certifying positive invariance \cite{ames2019control,clark2021verification}. However, SOS methods are restricted to polynomial systems and limit the expressive power of the resulting certificates. To address these limitations, barrier certificates defined by neural networks—known as neural barrier certificates—have been introduced \cite{dawson2023safe,zhao2020synthesizing,qin2021learning,abate2021fossil,zhao2021learning,liu2023safe}. Leveraging their universal approximation capability, neural barrier certificates have demonstrated promising performance in applications such as robot control \cite{dawson2022safe,xiao2023barriernet}. 
Nonetheless, due to the approximate nature of neural networks, they may fail to guarantee positive invariance. Therefore, verification methods are essential to certify the correctness of learned neural barrier certificates.

This paper focuses on the verification and falsification of neural barrier certificates using Rectified Linear Unit (ReLU) activation functions, due to their widespread use in the safe control community \cite{dawson2023safe, zhao2021synthesizing, mathiesen2022safety}.
However, ReLU neural barrier certificates are not differentiable, making traditional methods that rely on Lie derivative conditions inapplicable \cite{dai2017barrier, ames2019control}.
Under the assumption that the derivative of the ReLU activation function is the Heaviside step function (an assumption that lacks mathematical rigor for verifying positive invariance; see the appendix for details),
some works
\cite{zhao2022verifying,hu2024verification} over-approximate the possible values of the Lie derivative and then verify the Lie derivative condition over these over-approximations, leveraging either mixed integer programming \cite{zhao2022verifying} or symbolic bound propagation \cite{hu2024verification}.

Recently, \cite{zhang2023exact} proposed a necessary and sufficient condition for verifying positive invariance based on the Bouligand tangent cone, 
which was further used to synthesize ReLU neural barrier certificates in \cite{zhang2024seev}. 
Their condition avoids assuming that the derivative of the ReLU activation is the Heaviside step function. 
However, The most computationally intensive step in their method is enumerating all possible intersection combinations of linear regions intersecting the boundary of the 0-superlevel set of the ReLU neural barrier certificate.
This enumeration procedure has exponential complexity with respect to the number of linear regions containing boundary points (see Remark \ref{remark: compare to hinge}). 
Additionally, 
the Bouligand tangent cone condition  cannot be encoded exactly in optimization problems due to the presence of strict inequalities (see Remark \ref{remark: Bouligand cannot opt}). As a result, the condition must be relaxed to a sufficient one to enable incorporation into optimization formulations.

In summary, all existing methods for verifying ReLU neural barrier certificates rely on derivative assumption or sufficient conditions, which can be overly conservative and lead to false negatives in practical applications. Moreover, none of these methods can be used for falsification, leaving a critical gap in real-world deployment.
In this paper, we propose a novel necessary and sufficient condition for certifying the positive invariance of 0-superlevel sets of continuous piecewise linear functions (CPLFs) under continuous-time systems. This condition is applicable not only to ReLU neural barrier certificates but also to other networks with piecewise linear activation functions, such as leaky ReLU \cite{maas2013rectifier} and PReLU \cite{he2015delving}, since these networks are inherently CPLFs.
Our proposed condition states that the 0-superlevel set of a CPLF is positively invariant if and only if, in each valid linear region (see Definition \ref{def:Valid linear region}), the inner product between the region's linear coefficient vector and the vector field is non-negative. Compared to the condition in \cite{zhang2023exact}, our condition requires verification in significantly fewer regions, as it avoids enumerating all possible  intersection combinations of  linear regions that  intersect boundary. This reduction  makes the proposed condition more efficient and practical for real-world applications.


We propose a verification algorithm based on our necessary and sufficient condition. The algorithm begins by identifying an initial valid linear region using the Interval Bound Propagation (IBP) technique. Once such a region is found, a boundary propagation algorithm is employed to enumerate all neighboring valid linear regions. By iteratively applying this propagation step, the algorithm ensures that all valid linear regions are covered.
For each valid linear region, we can translate the proposed condition into Satisfiability Modulo Theories (SMT) and optimization problems for verifying and falsifying the ReLU neural barrier certificate. Finally,
Numerical experiments demonstrate the validity and effectiveness of our method in both verification and falsification.
The main contributions of this paper is summarized as follows.
\begin{itemize}
    \item 
    We propose a necessary and sufficient condition for certifying the positive invariance of 0-superlevel sets of CPLFs under continuous-time systems. Unlike \cite{zhao2022verifying,hu2024verification}, our condition does not rely on the assumption that the derivative of the ReLU activation is the Heaviside step function. Besides, compared to the condition in \cite{zhang2023exact}, it significantly reduces the computational complexity of implementation.
    \item  
    The proposed condition can be encoded into SMT and optimization problems, enabling both verification and falsification of ReLU neural barrier certificates. To the best of our knowledge, this is the first method capable of falsifying such certificates.
    
\end{itemize}

\section{Preliminaries}
\label{sec:pre}
\subsection{Notation}

$\mathbb{R}^{n}$ represents  $n$-dimensional real space; $\mathbb{R}^{m\times n}$ represents space of $m \times n$ real matrices;
 $\mathbb{N}_{[m,n]}$ represents the non-negative integers in $[m,n]$. Vectors and matrices are denoted as boldface lowercase and uppercase respectively. For a vector $\bm{x}$, $\bm{x}(i)$ represents its $i$-th entry; $\Vert \bm{x} \Vert$  represents its norm. $ \bm{x} \cdot \bm{y} $ represents inner product of vectors $\bm{x}$ and $\bm{y}$. For a matrix $\bm{M}$,  $\bm{M}(i)$  represents its  $i$-th row vector; $\text{rows}(\bm{M})$ and $\text{cols}(\bm{M})$ denote the number of its rows and columns respectively; $\rank(\bm{M})$ represents the rank of $\bm{M}$; $[\bm{M}; \bm{x}^\top]$ represents adding  $\bm{x}$ to  the last row  of $\bm{M}$. $\bm{0}$ (or $\bm{1}$) represents the vector (or matrix) whose entries are all zero (or one) with appropriate dimensions in the context. For a set $\mathcal{S}$, its complement, interior, closure, boundary, cardinality and  power set   are denoted by $\mathcal{S}^c, \Int\mathcal{S}, \overline{\mathcal{S}}$, $\partial \mathcal{S}$, $\vert\mathcal{S}\vert$ and $2^{\mathcal{S}}$, respectively. For two sets $\mathcal{S}_1, \mathcal{S}_2$, 
$\mathcal{S}_1 \setminus \mathcal{S}_2$ represents the set $\{s:s\in \mathcal{S}_1 \wedge s\notin \mathcal{S}_2\}$. $B(\bm{x},\delta)= \{\bm{x}^\prime \in \mathbb{R}^n:\Vert \bm{x}^\prime -\bm{x} \Vert \leq \delta\}$ represents the closed $\delta$-ball around the vector $\bm{x}$.

\subsection{ReLU Neural Network}
\label{subsec: ReLU}
We introduce notations to describe a neural network (NN) with $L$ hidden layers, where the $i$-th layer contains $M_i$ neurons. Let $\bm{x}\in \mathbb{R}^n$  denote the input to the network, $z_{i j}$ the output of the $j$-th neuron in the $i$-th layer, and $y$ the one-dimensional network output. 
We use $\bm{z}_i$ to represent the vector of neuron outputs in the  $i$-th layer. The outputs are computed as 
$$
z_{i j}=\left\{\begin{array}{ll}
\sigma\left(\bm{w}_{ij}^\top \bm{x}+b_{i j}\right), & i=1 \\
\sigma\left(\bm{w}_{ij}^\top \bm{z}_{i-1}+b_{i j}\right), & 
2\leq i \leq L
\end{array}~y=\bm{\omega}^\top \bm{z}_L+\phi\right.$$
where $\sigma: \mathbb{R} \rightarrow \mathbb{R}$ is the activation function. The input to $\sigma$ is the pre-activation value to the neuron:
for the $j$-th neuron in the first layer, this value is given by $\bm{w}_{1j}^\top \bm{x} + b_{1j}$; for the $j$-th neuron in the $i$-th hidden layer ($i > 1$), it is given by $\bm{w}_{ij}^\top \bm{z}_{i-1} + b_{ij}$. Here, $\bm{w}_{ij} \in \mathbb{R}^n$ when $i = 1$, and $\bm{w}_{ij} \in \mathbb{R}^{M_{i-1}}$ when $i > 1$.
Throughout this paper, we assume  $\sigma$ is the ReLU function $\sigma(z)=\max \{0, z\}$. 
The final output of the network is given by $y=\bm{\omega}^\top \bm{z}_L+\phi$, where $\bm{\omega} \in \mathbb{R}^{M_L}$ and $\phi \in \mathbb{R}$. A neuron is said to be \textit{activated} by an input $\bm{x}$ if its pre-activation value is non-negative, and \textit{inactivated} if it is non-positive. If the pre-activation value is exactly zero, the neuron is considered both activated and inactivated. 


An \textit{activation indicator} is an $L$-tuple $\mathscr{C} = \langle\bm{s}_1, \bm{s}_2, \ldots, \bm{s}_L\rangle$, where each $\bm{s}_i = (s_{i1}, s_{i2}, \ldots, s_{iM_i})^\top$ is a binary vector of length $M_i$. Each entry $s_{ij} \in \{0, 1\}$ indicates whether the $j$-th neuron in the $i$-th layer is inactivated ($0$) or activated ($1$).


For a given activation indicator $\mathscr{C}$, if an input $\bm{x}$ activates $\mathscr{C}$, the pre-activation values of all neurons, as well as the overall network output, are affine functions of $\bm{x}$. The corresponding affine mapping is determined by $\mathscr{C}$ as follows.
For the first layer, we define:
$$
\bar{\bm{w}}_{1 j}(\mathscr{C})=\left\{\begin{array}{ll}
\bm{w}_{1 j}, & s_{1j} = 1 \\
0, & s_{1j} = 0
\end{array} \quad \bar{b}_{1 j}(\mathscr{C})= \begin{cases}b_{1 j}, & s_{1j} = 1 \\
0, & s_{1j} = 0\end{cases}\right.
$$
So that the output of the $j$-th neuron in the first layer is given by $\bar{\bm{w}}_{1 j}(\mathscr{C})^\top x+\bar{b}_{1 j}(\mathscr{C})$. 
We recursively define $\bar{\bm{w}}_{ij}(\mathscr{C})$ and $\bar{b}_{ij}(\mathscr{C})$ for $i > 1$ by letting $\overline{\bm{W}}_i(\mathscr{C})$ be the matrix whose columns are $\bar{\bm{w}}_{i1}(\mathscr{C}), \ldots, \bar{\bm{w}}_{iM_i}(\mathscr{C})$, and setting:
\begin{equation*}
    \begin{aligned}
 &\bar{\bm{w}}_{i j}(\mathscr{C})=\left\{\begin{array}{ll}
\overline{\bm{W}}_{i-1}(\mathscr{C}) \bm{w}_{ij}, & s_{ij} = 1 \\
0, & s_{ij} = 0
\end{array}\right.   \\ &\bar{b}_{i j}(\mathscr{C})= \begin{cases}\bm{w}_{ij}^\top \bar{\bm{b}}_{i-1}(\mathscr{C})+b_{i j}, & s_{ij} = 1\\
0, & s_{ij} = 0\end{cases}       
    \end{aligned}
\end{equation*}
where $\bar{\bm{b}}_i(\mathscr{C})$ is the vector of bias terms $\bar{b}_{ij}(\mathscr{C})$ for $j = 1, \ldots, M_i$.
We define $\bm{w}(\mathscr{C})=\overline{\bm{W}}_L(\mathscr{C}) \bm{\omega}$ and $b(\mathscr{C})=\bm{\omega}^\top \bar{\bm{b}}_L(\mathscr{C})+\phi$. Based on these notations, if the input $\bm{x}$ activates the activation indicator $\mathscr{C}$, the output of each neuron and the final network output are given by:
$z_{i j}=\bar{\bm{w}}_{i j}(\mathscr{C})^\top \bm{x}+\bar{b}_{i j}(\mathscr{C})$ and $y=\bm{w}(\mathscr{C})^\top \bm{x}+b(\mathscr{C})$.
Next, we characterize the activated region corresponding to a given activation indicator $\mathscr{C}$.

\begin{lemma}[\cite{zhang2023exact}]
Let $\mathcal{X}(\mathscr{C})$ denote the set of inputs  that activate a particular set of neurons represented by the activation indicator $\mathscr{C}$. For notational consistency, we define $\overline{\bm{W}}_0(\mathscr{C})$ as the identity matrix and $\bar{\bm{b}}_0(\mathscr{C})$ as the zero vector. Then
\begin{equation*}
    \begin{gathered}
\mathcal{X}(\mathscr{C})=\bigcap_{i=1}^L\Bigg(\bigcap_{j=1}^{M_i}\{\bm{x}\in \mathbb{R}^n: \bm{w}_{ij}^\top\left(\overline{\bm{W}}_{i-1}(\mathscr{C})^\top \bm{x} \right.\\ \left. +\bar{\bm{b}}_{i-1}\right)+b_{i j} \geq 0, \bm{s}_i(j)=1\}\cap \bigcap_{j=1}^{M_i}\left\{\bm{x}\in \mathbb{R}^n:\right.\\ \left.\bm{w}_{ij}^\top\left(\overline{\bm{W}}_{i-1}(\mathscr{C})^\top \bm{x} +\bar{\bm{b}}_{i-1}\right)+b_{i j} \leq 0, \bm{s}_i(j)=0\right\}\Bigg).        
    \end{gathered}
\end{equation*}
\end{lemma}

The activated region $\mathcal{X}(\mathscr{C})$ forms a polyhedron. Let  $\mathcal{I}$ denote the set of all possible activation indicators. With the above notations, the ReLU neural network can be expressed as a \textbf{continuous piecewise linear function (CPLF)}:
\begin{equation}
    y=\bm{w}(\mathscr{C})^\top \bm{x}+b(\mathscr{C}), \bm{x} \in \mathcal{X}(\mathscr{C}), \mathscr{C} \in \mathcal{I}.
\end{equation}
Note that  $\mathcal{I}$ is a finite set, as the number of activated regions is no more than $2^{\sum_{i=1}^{L} M_i}$ \cite{montufar2014number}.

\subsection{Positive Invariance and Barrier Certificate}
In this paper we consider the continuous-time system
\begin{equation}
\label{system}
    \dot{\bm{x}}=\bm{f}(\bm{x}),
\end{equation}
with $\bm{x} \in \mathbb{R}^n$ and $\bm{f}:\mathbb{R}^n\rightarrow \mathbb{R}^n$ locally Lipschitz. For any initial condition $\bm{x}_0\in \mathbb{R}^n$, there exists a maximal time interval of existence $I(\bm{x}_0)=[0,\tau_{\max})$ such that $\bm{\phi}_{\bm{x}_0}: I(\bm{x}_0)\rightarrow \mathbb{R}^n$ is the unique solution to system \eqref{system}, where $\bm{\phi}_{\bm{x}_0}(0)=\bm{x}_0$ and $\tau_{\max}$ is the explosion time with $\lim_{t\rightarrow  \tau_{\max}} \Vert \bm{\phi}_{\bm{x}_0}(t) \Vert=+\infty$.

\begin{definition}[Positive invariance]
    A set $\mathcal{C} \subseteq \mathbb{R}^n$ is positively invariant for system \eqref{system} if for all $\bm{x}_0 \in \mathcal{C}$ and all $t \in I(\bm{x}_0)$, the corresponding trajectory satisfies $\bm{\phi}_{\bm{x}_0}(t) \in \mathcal{C}$. 
\end{definition}

\begin{definition}[Barrier certificate]
\label{def: Barrier certificate}
  Given the system \eqref{system}, let the
initial set be $\mathcal{S}_I = \{\bm{x}\in\mathbb{R}^n: h_I(\bm{x}) > 0\}$ and the unsafe set be $\mathcal{S}_U = \{\bm{x}\in\mathbb{R}^n: h_U(\bm{x}) > 0\}$, 
where $h_I, h_U$ are continuous functions, and 
 both $\mathcal{S}_I, \mathcal{S}_U$ are nonempty and connected.
A barrier certificate for system \eqref{system} is a continuous function $h :\mathbb{R}^n \rightarrow \mathbb{R}$ whose 0-superlevel set $\mathcal{C} = \{\bm{x} \in \mathbb{R}^n: h(\bm{x}) \geq 0\}$ satisfies the following conditions:
\begin{enumerate}
    \item \textbf{Initial set condition}: $ \mathcal{S}_I \subset \mathcal{C}$.
    \item \textbf{Unsafe set condition}: $ \mathcal{S}_U \cap \mathcal{C} = \emptyset$.
    \item \textbf{Positively invariant condition}: $\mathcal{C}$ is positively invariant under system \eqref{system}.
\end{enumerate}
\end{definition}
Once a barrier certificate is found, it guarantees that all trajectories starting from the initial set will never enter the unsafe set.
A \textbf{ReLU neural barrier certificate} refers to a barrier certificate  represented by a neural network with ReLU activation functions.
The objective of this paper is to verify or falsify a given ReLU neural barrier certificate.

\begin{remark}
    In fact, it suffices for a single connected component of $\mathcal{C}$ to satisfy the three conditions to ensure that all trajectories starting from the initial set will never enter the unsafe set. In the verification algorithm proposed in Sect.\ref{sec: Verification Algorithm}, the boundary propagation algorithm is employed to identify the complete boundary of such a connected component of $\mathcal{C}$.
\end{remark}

\section{Tangent Cones and Invariance Conditions}
\label{sec: Set Invariant Theories on Boundary}
Since a ReLU neural network is essentially a CPLF, this section investigates  the necessary and sufficient conditions for the positively invariant condition to the 0-superlevel set defined  by a CPLF.
We begin by reviewing some significant theorems about positively invariance and related foundational concepts.

\begin{definition}[Distance]
    Given a vector space $X$ with norm $\|\cdot\|$, the distance between two points $\bm{x}_1, \bm{x}_2 \in X$ is  
$d(\bm{x}_1, \bm{x}_2)=\Vert\bm{x}_1-\bm{x}_2\Vert$; the distance between a set $\mathcal{S} \subset X$ and a point $\bm{x} \in X$ is        
$d(\bm{x}, \mathcal{S})=\inf _{\bm{y}\in \mathcal{S}}\Vert\bm{x}-\bm{y}\Vert.$
\end{definition}


\begin{definition}[Tangent cones \cite{clarke2008nonsmooth,aubin2009set}]
\label{def: tangent cone}
    Let $\mathcal{S}$ be a closed subset of the Banach space $X$. 
    \begin{enumerate}
        \item The Bouligand  tangent cone or contingent cone to $\mathcal{S}$ at $\bm{x}$, denoted $T_\mathcal{S}^B(\bm{x})$, is defined as follows:
        \begin{equation*}
    T_\mathcal{S}^B(\bm{x})\triangleq\left\{\bm{v}\in X~\Big|~ \liminf _{t \rightarrow 0^{+}} \frac{d(\bm{x}+t \bm{v}, \mathcal{S}) }{t}=0\right\}.
\end{equation*}
\item The Clarke tangent cone or circatangent cone to $\mathcal{S}$ at $\bm{x}$, denoted $T_\mathcal{S}^C(\bm{x})$, is defined as follows:
        \begin{equation*}
    T_\mathcal{S}^C(\bm{x})\triangleq\left\{\bm{v}\in X~\bigg|~ \lim _{t \rightarrow 0^{+}, \bm{x}^\prime \stackrel{\mathcal{S}}{\rightarrow} \bm{x}} \frac{d(\bm{x}^\prime+t \bm{v}, \mathcal{S}) }{t}=0\right\},
\end{equation*}
where $\bm{x}^\prime \stackrel{\mathcal{S}}{\rightarrow} \bm{x}$ means the convergence is in $\mathcal{S}$.
    \end{enumerate}
\end{definition}

\begin{figure}[!htbp]
\begin{subfigure}{.24\textwidth}
  \includegraphics[width= \textwidth]{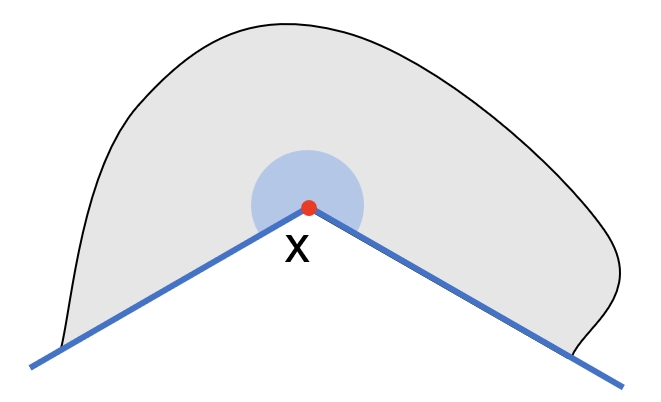}
  \subcaption*{(a) Bouligand tangent cone}
\end{subfigure}%
\begin{subfigure}{.24\textwidth}
  \includegraphics[width= \textwidth]{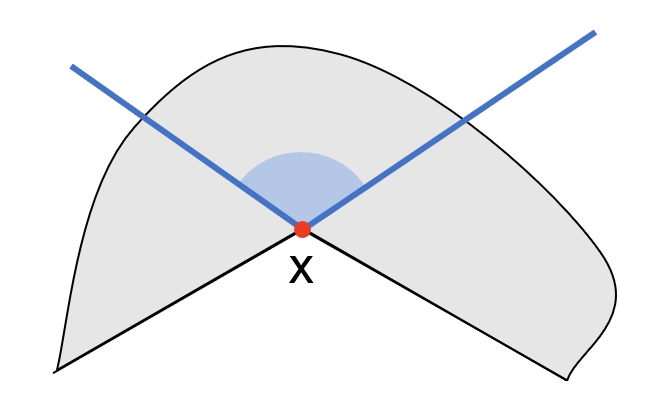}
  \subcaption*{(b) Clarke tangent cone}
\end{subfigure}%
\caption{Illustration of tangent cones at a nonsmooth boundary point}
\end{figure}


Both  tangent cones in Definition \ref{def: tangent cone} are closed cones. If $\bm{x} \in \Int{\mathcal{S}}$, then $T_\mathcal{S}^B(\bm{x})=T_\mathcal{S}^C(\bm{x}) = X$. If $\bm{x} \in \mathcal{S}^c$, then $T_\mathcal{S}^B(\bm{x})=T_\mathcal{S}^C(\bm{x}) = \emptyset$. Therefore, tangent cones are nontrivial only on the boundary $\partial \mathcal{S}$.

The following theorem is quoted from \cite{clarke2008nonsmooth} with some modifications for the context of ordinary differential equations instead of differential inclusions.  

\begin{theorem} [{\cite[Theorem 3.8 in Chapter 4]{clarke2008nonsmooth}}]
\label{Nagumo’s theorem}
     Consider the system \eqref{system} and let $\mathcal{S} \subset \mathbb{R}^n$ be a closed set, then the following assertions are equivalent:
     \begin{enumerate}[itemsep=3pt,topsep=2pt,parsep=0pt]
         \item [a.] $\mathcal{S}$ is positively invariant for the system \eqref{system};
         \item [b.] for all $\bm{x} \in \partial \mathcal{S}$, $\bm{f}(\bm{x}) \in T_{\mathcal{S}}^B(\bm{x})$;
         \item [c.] for all $\bm{x} \in \partial\mathcal{S}$, $\bm{f}(\bm{x}) \in T_{\mathcal{S}}^C(\bm{x})$.
    \end{enumerate}
\end{theorem}

 For a CPLF $h:\mathbb{R}^n \rightarrow \mathbb{R}$, its expression can be written as 
    \begin{equation}
    \label{eq:linear function expression}
        h(\bm{x}) = \bm{w}_i^\top \bm{x} + b_i, \bm{x} \in \mathcal{X}_i, i =1, \cdots, N,
    \end{equation}
    where $\mathcal{X}_i =\{\bm{x}: \bm{A}_i \bm{x} \leq \bm{d}_i\}$ is a polyhedron (the background of polyhedron is provided in the appendix)  of full dimension, i.e., $\dim(\mathcal{X}_i) = n$, referred to as a \textit{linear region}. 
   The collection \( \{ \mathcal{X}_i \}_{i=1}^N \) partitions the entire input space \( \mathbb{R}^n \), i.e., \( \bigcup_{i=1}^N \mathcal{X}_i = \mathbb{R}^n \), and for all \( i \neq j \), \( \dim(\mathcal{X}_i \cap \mathcal{X}_j) \leq n - 1 \).
We say \( \mathcal{X}_i \) and \( \mathcal{X}_j \) are adjacent if \( \dim(\mathcal{X}_i \cap \mathcal{X}_j) = n - 1 \). In such a case, there exists a hyperplane \( \mathcal{H} \) such that \( \mathcal{X}_i \cap \mathcal{X}_j \subseteq \mathcal{H} \), and \( \mathcal{H} \cap \mathcal{X}_i \) and \( \mathcal{H} \cap \mathcal{X}_j \) are facets of \( \mathcal{X}_i \) and \( \mathcal{X}_j \), respectively. 

    
    
In this paper, we adopt  the following assumption, which is also emphasized in \cite{ames2016control,ames2019control}.
    \begin{assumption}
    \label{ass1}
For a candidate barrier certificate $h$, its   0-superlevel set $\mathcal{C} = \{\bm{x} \in \mathbb{R}^n: h(\bm{x}) \geq 0\}$  satisfies:
        $\partial \mathcal{C} = \{\bm{x}\in \mathbb{R}^n: h(\bm{x}) = 0\}$, 
        $\Int \mathcal{C} = \{\bm{x}\in \mathbb{R}^n: h(\bm{x}) > 0\}$;
         $\mathcal{C}$ contains interior points and is a regular closed set, i.e., $\Int \mathcal{C} \neq \emptyset, \overline{\Int{\mathcal{C}}}=\mathcal{C}$.
\end{assumption}



   For ease of presentation, 
    we define valid linear regions as follows.
    
    \begin{definition}[Valid linear region]
    \label{def:Valid linear region}
        A linear region $\mathcal{X}_i$ is said to be \textbf{valid} if it is $n$-dimensional and, for the associated hyperplane $\mathcal{H}_i = \{\bm{x} \in \mathbb{R}^n : \bm{w}_i^\top \bm{x} + b_i = 0\}$, either $\mathcal{H}_i \cap \Int \mathcal{X}_i \neq \emptyset$ or $\mathcal{H}_i \cap \mathcal{X}_i$ is a facet of $\mathcal{X}_i$.
    \end{definition}

The following two propositions respectively characterize the Bouligand tangent cone and the Clarke tangent cone to the set $\mathcal{C}$. 
\begin{proposition}
\label{lemma: linear function b cone}
   Given a CPLF $h$ as defined in~\eqref{eq:linear function expression}, consider its 0-superlevel set $\mathcal{C} = \{\bm{z} \in \mathbb{R}^n : h(\bm{z}) \geq 0\}$. Under Assumption~\ref{ass1}, the following assertions hold for any $\bm{x} \in \partial \mathcal{C}$:
 
\begin{itemize}
    \item[1.] if $\bm{x} \in \Int{\mathcal{X}_i},  i = 1,\cdots,  N$, then
    \begin{equation}
    \begin{aligned}
    T_\mathcal{C}^B(\bm{x}) = 
            \{\bm{v}\in \mathbb{R}^n: \bm{w}_i^\top \bm{v} \geq 0\};
    \end{aligned}
        \end{equation}

        \item[2.] if $\bm{x} \in \bigcap_{k=1}^{m} {\mathcal{X}_{i_k}} $ and $\bm{x} \not\in \overline{\mathbb{R}^n \setminus\bigcup_{k=1}^{m} \mathcal{X}_{i_k}}, i_1,\cdots, i_m \in \{1,2,\cdots, N\}$, then 
   \begin{equation}
   \label{subeq: insection b cone}
       \begin{aligned}
        &T_\mathcal{C}^B(\bm{x}) =\bigcup_{k=1}^m\{\bm{v}\in \mathbb{R}^n:
            \bigwedge_{j \in E} \bm{A}_{i_k}(j) \bm{v} \leq 0 \wedge \bm{w}_{i_k}^\top \bm{v} \geq 0\},
       \end{aligned}
   \end{equation} 
   \begin{equation}
   \label{subeq: insection b cone subset} 
              T_\mathcal{C}^B(\bm{x}) \supset  \{\bm{v}\in \mathbb{R}^n: \bigwedge_{k\in I}
            \bm{w}_{i_k}^\top \bm{v} \geq 0\}, 
   \end{equation}
        where the set $E,I$ are defined as $E \triangleq\{j \in \{1,\cdots, \rows(\bm{A})\}: \bm{A}_{i_k}(j) \bm{x} = \bm{d}_{i_k}(j)\}, I \triangleq\{k \in \{1,\cdots,m\}: \mathcal{X}_{i_k} \text{ is a valid linear region}\}$.
        \item[3.] $\partial \mathcal{C} \subset \bigcup_{l \in J} \mathcal{X}_l$, where $J \triangleq\{l \in \{1,\cdots,N\}: \mathcal{X}_{l} \text{ is a valid linear region}\}$.
\end{itemize}

\end{proposition}

\begin{proposition}
\label{lemma: linear function condition}
Given a CPLF $h$ as defined in~\eqref{eq:linear function expression}, let $\mathcal{C} = \{\bm{z} \in \mathbb{R}^n : h(\bm{z}) \geq 0\}$ denote its 0-superlevel set. Under Assumption~\ref{ass1}, the Clarke tangent cone to $\mathcal{C}$ at any point $\bm{x} \in \partial \mathcal{C}$ is given by
    \begin{equation}
    \label{eq: linear tangent cone}
        T_\mathcal{C}^C(\bm{x}) = \begin{cases}
            \{\bm{v}\in \mathbb{R}^n: \bm{w}_i^\top \bm{v} \geq 0\},~\text{if} \\~~\bm{x} \in \Int{\mathcal{X}_i},  i = 1,\cdots, N \\
            \{\bm{v}\in \mathbb{R}^n: \bigwedge_{k\in I}
            \bm{w}_{i_k}^\top \bm{v} \geq 0\},~\text{if}
     \\~~\bm{x} \in \bigcap_{k=1}^{m} {\mathcal{X}_{i_k}} \wedge \bm{x} \not\in \overline{\mathbb{R}^n \setminus\bigcup_{k=1}^{m} \mathcal{X}_{i_k}}
        \end{cases}
        \end{equation}
        where $i_1,\cdots, i_m \in \{1,2,\cdots, N\}, I =\{k \in \{1,\cdots,m\}: \mathcal{X}_{i_k} \text{is a valid linear region}\}$.
\end{proposition}


Based on the equivalence of assertions (a) and (c) in Theorem~\ref{Nagumo’s theorem}, the necessary and sufficient condition for the positive invariance of the 0-superlevel set of a CPLF under system~\eqref{system} is established in Theorem~\ref{thm: linear function invariance}.
\begin{theorem}
    \label{thm: linear function invariance}
         Consider the system~\eqref{system} and the set $\mathcal{C} = \{\bm{x} \in \mathbb{R}^n : h(\bm{x}) \geq 0\}$, where $h$ is a CPLF as defined in~\eqref{eq:linear function expression}. Under Assumption~\ref{ass1}, the set $\mathcal{C}$ is positively invariant for the system~\eqref{system} if and only if, for each valid linear region $\mathcal{X}_i$, 
         \begin{equation}
         \label{eq: linear invariance condition}
             \bm{w}_{i}^\top \bm{f}   (\bm{x}) \geq 0, \forall \bm{x} \in \partial \mathcal{C} \cap \mathcal{X}_{i}.
         \end{equation}
\end{theorem}

Recall that a ReLU neural network is essentially a CPLF. By applying Theorem~\ref{thm: linear function invariance} and using the notations introduced in Section~\ref{subsec: ReLU}, we derive the necessary and sufficient condition for the positive invariance of the 0-superlevel set represented by a ReLU neural network under system~\eqref{system}.

\begin{theorem}
    \label{thm: ReLU invariance}
         Consider the system~\eqref{system} and the set $\mathcal{C} = \{\bm{x} \in \mathbb{R}^n: h(\bm{x}) \geq 0\}$, where $h$ is a ReLU neural network as described in Section~\ref{subsec: ReLU}. Under Assumption~\ref{ass1}, the set $\mathcal{C}$ is positively invariant for system~\eqref{system} if and only if, for each activation indicator $\mathscr{C}$ such that 
         $\mathcal{X}(\mathscr{C})$ is a valid linear region,
         \begin{equation}
         \label{eq: relu invariance condition}
             \bm{w}(\mathscr{C})^\top \bm{f}(\bm{x}) \geq 0, \forall \bm{x} \in \partial \mathcal{C} \cap \mathcal{X}(\mathscr{C}).
         \end{equation}
\end{theorem}

\begin{remark}
\label{remark: compare to hinge}
    In fact, with the equivalence of assertions (a) and (b) in Theorem \ref{Nagumo’s theorem}, one can also derive a necessary and sufficient condition for the positive invariance of the 0-superlevel set of a ReLU neural network, as investigated in \cite{zhang2023exact}. However, compared to our condition in Theorem \ref{thm: ReLU invariance}, the condition in \cite{zhang2023exact} has to additionally enumerate all possible intersection combinations of activated regions $\mathcal{X}(\mathscr{C})$ that intersect the boundary $\partial \mathcal{C}$, which significantly increases computational complexity.  For instance, if there are $n$ activated regions $\mathcal{X}(\mathscr{C}_1), \cdots, \mathcal{X}(\mathscr{C}_n)$ with a nonempty intersection $\bigcap_{k=1}^n \mathcal{X}(\mathscr{C}_k) \neq \emptyset$,  the total number of intersection combinations among these regions  is $\sum_{k=1}^n {n\choose k} =2^n-1$.
    However, according to Definition \ref{def:Valid linear region}, the number of valid linear regions that must be enumerated in our method is fewer than $n$.
\end{remark}
\section{Verification Algorithm}
\label{sec: Verification Algorithm}

In this section we present  our verification algorithm based on the necessary and sufficient condition in Theorem \ref{thm: ReLU invariance}.
 The algorithm proceeds in three steps: we first search for an initial valid linear region (Sect. \ref{subsect: Enumerating linear regions Intersecting Boundary}); using this initial region, the boundary propagation algorithm enumerates all valid linear regions, i.e., compute $\mathcal{A}\triangleq\{\mathscr{C}: 
 \mathcal{X}(\mathscr{C}) \text{ is a valid linear region}\}$ (Sect. \ref{subsect: Boundary Propagation Algorithm}); finally, for each valid activation indicator $\mathscr{C} \in \mathcal{A}$, we verify whether $\bm{w}(\mathscr{C})^\top \bm{f}(\bm{x}) \geq 0$ holds for all $\bm{x} \in \partial \mathcal{C} \cap \mathcal{X}(\mathscr{C})$ (Sect. \ref{subsec: Verification in Each Activation Region}). If this condition holds for every $\mathscr{C} \in \mathcal{A}$, then $\mathcal{C}$ is positively invariant for system \eqref{system}; otherwise, if any $\mathscr{C}^\prime \in \mathcal{A}$ violates the condition, $\mathcal{C}$ is not positively invariant. Figure \ref{fig:enumeration_algorithm} illustrates the procedure of enumerating all valid linear regions.

\begin{figure}[htbp]
    \centering
    \includegraphics[width=\linewidth]{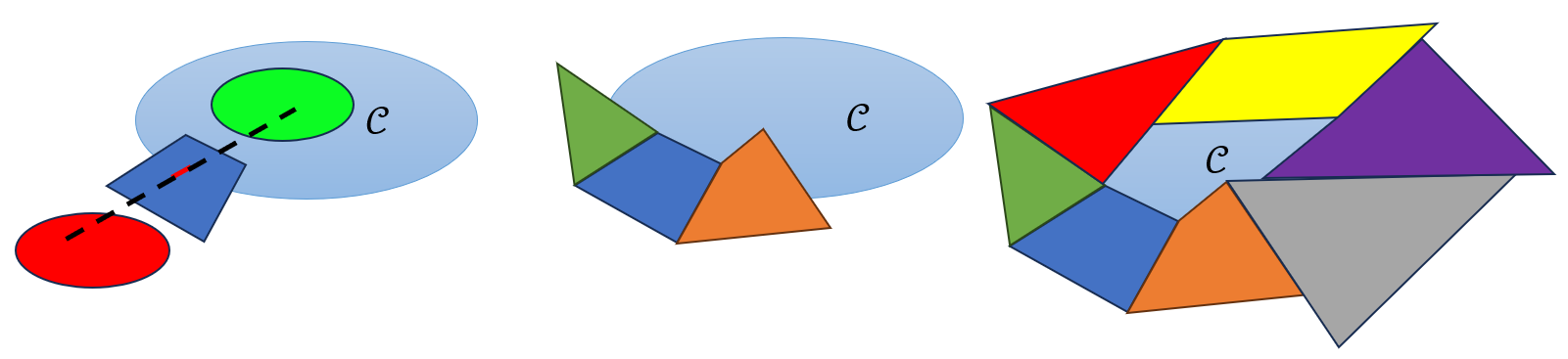}
    \caption{
    Illustration  of enumerating valid linear regions:
    We begin by randomly selecting two points---one from the initial set (colored in green) and one from the unsafe set (colored in red)---and iteratively shrinking the segment connecting them until an initial valid linear region is found. 
The boundary propagation algorithm then expands from this region to identify neighboring valid linear regions. 
By iteratively applying
this propagation step, the algorithm eventually enumerates all valid linear regions.
    }
    \label{fig:enumeration_algorithm}
\end{figure}

\subsection{Searching for an Initial Valid Linear Region}
\label{subsect: Enumerating linear regions Intersecting Boundary}

We use the following procedure to search for an initial valid linear region.

\begin{itemize}
    \item[1.] Randomly select points until  two points $\bm{x}_1,\bm{x}_2$ are found such that $h(\bm{x}_1)h(\bm{x}_2) <0$. For example,  $\bm{x}_1$ can be sampled from the unsafe set and $\bm{x}_2$ from the initial set.
    \item[2.] Utilize the method of the bisection to shrink the length of the line segment $\{t\bm{x}_1 +(1-t)\bm{x}_2: 0\leq t \leq 1\}$, until $\Vert \bm{x}_1 -\bm{x}_2 \Vert \leq \epsilon$, where $\epsilon>0$ is a predefined threshold.  Specifically, suppose initially $h(\bm{x}_1)<0, h(\bm{x}_2)>0$, if $h(\frac{\bm{x}_1+\bm{x}_2}{2}) <0$, then $\bm{x}_1 := \frac{\bm{x}_1+\bm{x}_2}{2}$, else $\bm{x}_2 := \frac{\bm{x}_1+\bm{x}_2}{2}$.
    \item[3.] Compute the interval hull $\Inthull(\{\bm{x}_1,\bm{x}_2\})$ of $\bm{x}_1$ and $\bm{x}_2$, 
    and use Interval Bound Propagation (IBP) to determine  the  candidate activation indicator $\mathscr{C}^\prime$. For each neuron $j$ in layer $i$ with input interval $[l_1, l_2]$:
    \begin{itemize}
        \item if $l_1>0$, set $\bm{s}_i(j) := 1$;
        \item if $l_2<0$, set $\bm{s}_i(j) := 0$;
         \item else, set $\bm{s}_i(j) := -1$.
        
    \end{itemize}
    \item[4.] Generate all feasible activation indicators $\mathscr{C}$ from $\mathscr{C}^\prime$ by replacing  each $-1$ in $\bm{s}_i(j)$ with both  1 and 0.
    \item[5.] 
    
 For each feasible activation indicator $\mathscr{C}$, check whether $\mathcal{X}(\mathscr{C})$ passes the \textbf{Valid Test} below. If a valid indicator is found, denote it by $\mathscr{C}_0$, and take the corresponding region $\mathcal{X}(\mathscr{C}_0)$ as the initial valid linear region.
\end{itemize}

\paragraph{Valid Test}
When $\mathcal{X}(\mathscr{C})$ is an $n$-dimensional valid linear region, it follows that $\dim(\partial \mathcal{C} \cap \mathcal{X}(\mathscr{C})) =n-1$. Therefore, the  validity of $\mathcal{X}(\mathscr{C})$ can be verified by checking whether this dimensional condition holds. Note that the redundant case where $\partial \mathcal{C} \cap \mathcal{X}(\mathscr{C}) = \mathcal{X}(\mathscr{C})$ can also be included, but this does not  affect the correctness of our method.

Given an activated region  of the form $\mathcal{X}(\mathscr{C}) =\{\bm{x} \in \mathbb{R}^n: \bm{A} \bm{x} \leq \bm{d}\}$, define the set $\mathcal{P} = \partial \mathcal{C} \cap \mathcal{X}(\mathscr{C})$ and let $\Tilde{\bm{A}} = [\bm{A};\bm{w}^\top(\mathscr{C});-\bm{w}^\top(\mathscr{C})], \Tilde{\bm{d}} = [\bm{d};-\bm{b}(\mathscr{C});\bm{b}(\mathscr{C})]$, then $\mathcal{P} =\{ \bm{x} \in \mathbb{R}^n: \Tilde{\bm{A}} \bm{x} \leq \Tilde{\bm{d}}\}$ is also a polyhedron. 
In a system of linear inequalities $\bm{A}\bm{x} \leq \bm{d}$, the inequality $\bm{A}(j)\bm{x} \leq \bm{d}(j), j \in \{1,\cdots, \rows(\bm{A})\}$, is called an \textit{implicit equality} if $\bm{A}(j)\bm{x} =\bm{d}(j) $ holds for all solutions of the system $\bm{A}\bm{x} \leq \bm{d}$. We denote by $\bm{A}^= \bm{x}\leq \bm{d}^=$ the system consisting of all implicit equalities within  $\bm{A}\bm{x} \leq \bm{d}$.

To identify the implicit equalities in $\mathcal{P} =\{ \bm{x} \in \mathbb{R}^n: \Tilde{\bm{A}} \bm{x} \leq \Tilde{\bm{d}}\}$, we solve the following pair of LPs for each row $j = 1, \cdots, \rows(\Tilde{\bm{A}})$:
  \begin{equation*}   
\begin{aligned}
        &\min_{\bm{x}} \Tilde{\bm{A}}(j)\bm{x}  &~~~~~~~~~~~~~~&\max_{\bm{x}} \Tilde{\bm{A}}(j)\bm{x}&\\
     &\text{s.t.}~ \Tilde{\bm{A}} \bm{x} \leq \Tilde{\bm{d}} &~~~~~~~~~~~~~~&\text{s.t.}~ \Tilde{\bm{A}} \bm{x} \leq \Tilde{\bm{d}}& 
\end{aligned}
\end{equation*}
If the optimal values of both LPs are equal, then $\Tilde{\bm{A}}(j) \bm{x} \leq \Tilde{\bm{d}}(j)$ is an implicit equality.
According to \cite[Theorem 4.17]{conforti2014integer}, the dimension of $\mathcal{P}$ satisfies $\dim(\mathcal{P}) = n -\rank(\Tilde{\bm{A}}^=)$. Therefore, if $\rank(\Tilde{\bm{A}}^=)=1$, then $\dim(\mathcal{P}) = n-1$.

\subsection{Boundary Propagation Algorithm}
\label{subsect: Boundary Propagation Algorithm}
Next we introduce the boundary propagation algorithm, which can enumerate all valid activation indicators, i.e., the set $\mathcal{A}\triangleq\{\mathscr{C}: 
\mathcal{X}(\mathscr{C}) \text{ is a valid linear region}\}$. This algorithm proceeds as follows:

\begin{enumerate}
    \item Initialize the set of valid activation indicators as $\mathcal{A} :=\{\mathscr{C}_0\}$ and the visited set as $\mathcal{B} := \emptyset$.
    \item Select an activation indicator $\mathscr{C}$ in $\mathcal{A} \setminus \mathcal{B}$ and add $\mathscr{C}$ to $\mathcal{B}$.
    Apply the method from \cite{marechal2017efficient} to remove the redundant constraints in $\mathcal{X}(\mathscr{C}) = \{\bm{x} \in \mathbb{R}^n: \bm{A} \bm{x} \leq \bm{d}\}$.
     For each $j = 1, \cdots, \rows(\bm{A})$, solve the  following LP feasibility problem
     to find the facet of $\mathcal{X}(\mathscr{C})$ which has the intersection with $\partial \mathcal{C}$.
                            \begin{equation*}   
                 \label{eq: find intersection}
\begin{cases}
              \bm{w}(\mathscr{C})^\top \bm{x} + b(\mathscr{C}) =0\\
         \bm{A}(j) \bm{x} = \bm{d}(j)\\
        \bm{A}(k) \bm{x} \leq \bm{d}(k), k \in \{1,\cdots, \rows(\bm{A})\}\setminus\{j\}
     \end{cases}
\end{equation*}

 If the  above LP  feasibility problem admits a feasible solution 
$\bm{x}^\star$, then enumerate all feasible activation indicators that are activated by 
$\bm{x}^\star$. For each indicator that passes the \textbf{Valid Test}, add it to the set 
$\mathcal{A}$.
\item The algorithm terminates when 
$\mathcal{A} = \mathcal{B}$, meaning all  valid activation indicators have been explored and no new ones have been found.

      \end{enumerate}

\begin{remark}
    When the system dimension is 2,  the number of facets of $\mathcal{X}(\mathscr{C})$ that intersect $\partial \mathcal{C}$ is at most 4. Therefore once four such facets have been identified, it is unnecessary to examine the remaining facets.
\end{remark}
  


For ReLU neural networks, the partitioning of the input space exhibits a special structure: every $n$-dimensional linear region shares a facet with each of its adjacent linear regions. This property arises from the fact that the partitioning is generated through a sequence  of hyperplanes, layer by layer \cite{raghu2017expressive, serra2018bounding}.
Moreover, since  $\partial \mathcal{C}$ is contained within the union of valid linear regions (see Proposition~\ref{lemma: linear function b cone}), if we further assume that $\partial \mathcal{C}$ is connected, then given an initial valid linear region, one can examine its facets to identify those  intersecting the boundary. By recursively expanding to adjacent regions via these shared  facets, we can  enumerate all valid linear regions. 
The following theorem formalizes this claim.
\begin{theorem}
\label{thm: bound connected}
    Assuming that $\partial \mathcal{C}$ is  connected, then the boundary propagation algorithm can identify  all  valid linear regions.
\end{theorem}

\subsection{Verification and Falsification in Each Valid Linear Region}
\label{subsec: Verification in Each Activation Region}

After enumerating all valid linear regions, we can verify or falsify the three conditions in Definition \ref{def: Barrier certificate} ---  positively invariant, initial set  and unsafe set conditions --- within each valid linear region.

\subsubsection{the Positively Invariant Condition}

The necessary and sufficient condition in Theorem \ref{thm: ReLU invariance} can be described as a quantified formula:
\begin{equation}
\label{eq:verificaiton Clarke}
\begin{aligned}
\forall  \mathscr{C} \in \mathcal{A}, \forall \bm{x} \, &\left( \bm{x} \in \partial \mathcal{C} \cap \mathcal{X}(\mathscr{C}) 
\rightarrow \bm{w}(\mathscr{C})^\top \bm{f}(\bm{x}) \geq 0 \right). 
\end{aligned}
\end{equation}

We can translate \eqref{eq:verificaiton Clarke} into an SMT problem that determines whether the following quantifier-free formula with Disjunctive Normal Form (DNF) is satisfiable (SAT) or not (UNSAT).

\begin{equation}
\label{eq:verificaiton Clarke SMT} 
\begin{aligned}
        \bigvee_{\mathscr{C}\in \mathcal{A}} 
             &\left(\bm{w}(\mathscr{C})^\top \bm{x} + b(\mathscr{C}) =0 \wedge
         \bm{x} \in \mathcal{X}(\mathscr{C}) \right.\\ &\left. \wedge \bm{w}(\mathscr{C})^\top \bm{f}(\bm{x}) < 0 \right).   
\end{aligned}
\end{equation}
If the SMT problem \eqref{eq:verificaiton Clarke SMT} is proven to be UNSAT, then the Boolean value of \eqref{eq:verificaiton Clarke} is true,  implying that  $\mathcal{C}$ is positively invariant under system \eqref{system}. 
When the system \eqref{system} is polynomial, 
the SMT \eqref{eq:verificaiton Clarke SMT} can be solved using an SMT solver that employs Cylindrical Algebraic Decomposition \cite{caviness2012quantifier}, such as Z3 \cite{de2008z3}. In this case,  verification by solving \eqref{eq:verificaiton Clarke SMT} is both sound and complete: an UNSAT return implies positive invariance of $\mathcal{C}$, while a SAT return indicates the falsification of positive invariance. When the system \eqref{system} is non-polynomial, 
 we can use dReal \cite{gao2013dreal}, an SMT solver that supports non-polynomial functions such as  trigonometric and exponential functions. dReal employs $\delta$-complete decision
procedures, returning either UNSAT or $\delta$-SAT, where $\delta$
is a user-specified numerical error bound. Therefore,  positive invariance verification using dReal is sound: when  dReal returns UNSAT for SMT \eqref{eq:verificaiton Clarke SMT},  $\mathcal{C}$ is positively invariant. However, a ``$\delta$-SAT'' result does not necessarily falsify positive invariance due to  potential numerical errors.

The quantified formula \eqref{eq:verificaiton Clarke} can also be described as a series of optimization problems:
for each valid linear region $\mathcal{X}(\mathscr{C}), \mathscr{C} \in \mathcal{A}$, we solve the following optimization.
\begin{equation}   
                 \label{eq: nonlinear verification}
\begin{split}
     &\min_{\bm{x}} \bm{w}(\mathscr{C})^\top \bm{f}(\bm{x})  \\
     \text{s.t.}~& \begin{cases}
              \bm{w}(\mathscr{C})^\top \bm{x} + b(\mathscr{C}) =0,\\
         \bm{x} \in \mathcal{X}(\mathscr{C}).
     \end{cases}
    \end{split}
\end{equation}
If  the optimal value of \eqref{eq: nonlinear verification} is non-negative, the positive invariance condition is satisfied  for $\mathscr{C}$. When all activation indicators in $\mathcal{A}$ pass the verification, we can deduce that $\mathcal{C}$ is positively invariant for system \eqref{system}; otherwise, it is not. 
When the system \eqref{system} is linear,  the optimization is an LP, therefore if the optimal value of \eqref{eq: nonlinear verification} exists, it can be obtained by the simplex method or interior-point method.
However, when the the system \eqref{system} is nonlinear, the optimization \eqref{eq: nonlinear verification} may be nonconvex, making it challenging to obtain the optimal value.  Whereas, we  we can utilize \eqref{eq: nonlinear verification} to falsify the positive invariance, Specifically, if there exists a valid linear region $\mathcal{X}(\mathscr{C}), \mathscr{C} \in \mathcal{A}$  such that \eqref{eq: nonlinear verification}  yields a negative feasible value,  then $\mathcal{C}$ is not positively invariant.

\begin{remark}
\label{remark: Bouligand cannot opt}
   The condition based on the Bouligand   tangent cone  cannot be directly translated into a series of optimization problems, in contrast to condition derived from the Clarke tangent cone. This limitation arises from the presence of strict inequalities --- specifically, the restriction ``for $\bm{x} \in \Int{\mathcal{X}_i}$ ($\bm{A}_i\bm{x} < \bm{d}_i)$'' in first assertion of Proposition \ref{lemma: linear function b cone} --- which are not permitted  in standard optimization formulations. The detailed explanation can be found in the appendix.
\end{remark}

\subsubsection{the Initial and Unsafe Set Conditions}
We can also encode the verification of the initial and unsafe set conditions as SMT and optimization problems like the positively invariant condition. To this end, we introduce the following propositions, which provide necessary and sufficient conditions for these conditions.
\begin{proposition}
\label{prop: initial condition}
      $\mathcal{S}_I \subset \mathcal{C}$  if and only if $\mathcal{S}_I \cap \partial \mathcal{C} =\emptyset$ and $\mathcal{S}_I \cap \Int \mathcal{C} \neq \emptyset$.
\end{proposition}

\begin{proposition}
\label{prop: unsafe condition}
 $\mathcal{S}_U \cap \mathcal{C} =\emptyset$ if and only if $\mathcal{S}_U \cap \partial \mathcal{C} =\emptyset$ and $\mathcal{S}_U \cap  \mathcal{C}^c \neq \emptyset$.
\end{proposition}

The  condition $\mathcal{S}_I \cap \partial \mathcal{C} =\emptyset$ can  also be described as a quantified formula:

\begin{equation}
\label{eq:verificaiton initial}
\begin{aligned}
    \forall  \mathscr{C} \in \mathcal{A}, \forall \bm{x} \, \left( \bm{x} \in \partial \mathcal{C} \cap \mathcal{X}(\mathscr{C}) \rightarrow h_I(\bm{x}) \leq 0 \right). 
\end{aligned}
\end{equation}
Similar to the positive invariance condition, the quantified formula \eqref{eq:verificaiton initial} can be encoded into SMT and optimization problems. For simplicity, we present the details in the appendix.

To verify  $\mathcal{S}_I \cap \Int \mathcal{C} \neq \emptyset$, we can randomly select a point $\bm{x} \in \mathcal{S}_I$ and check whether $h_I(\bm{x})>0$. If both conditions are satisfied, we conclude that 
$\mathcal{S}_I \subset \mathcal{C}$. Otherwise, the verification fails.

Note that the condition in Proposition \ref{prop: unsafe condition} for verifying the unsafe set condition 
$\mathcal{S}_U \cap \mathcal{C} =\emptyset$ is analogous to that used for verifying the initial set condition. Therefore, a similar verification procedure can be applied to the unsafe set condition as well.


\section{Experiments}
\label{sec: Experiments}
In this section, we evaluate the proposed method to demonstrate the validity  and effectiveness of our verification algorithm. Specifically, we conduct experiments on four systems: Arch3, Complex, Linear4d, and Decay. Among them, Linear4d is a linear system, while the others are nonlinear. System descriptions, experimental settings, and detailed  experimental results are provided in the appendix.

For linear systems, we only solve the optimization problem \eqref{eq: nonlinear verification} to check the positive invariance condition, as it reduces to an LP whose optimal value can be reliably obtained when it exists.
For nonlinear systems, we first solve the optimization \eqref{eq: nonlinear verification}. If there exists a valid linear region  for which the suboptimal value is negative, then the positive invariance condition is falsified. Otherwise, if all instances of \eqref{eq: nonlinear verification}  yield nonnegative suboptimal values for each $\mathscr{C} \in \mathcal{A}$, we proceed to solve the SMT problem \eqref{eq:verificaiton Clarke SMT} using the dReal solver \cite{gao2013dreal}. If dReal returns UNSAT for every region, positive invariance is verified. However, if it returns $\delta$-SAT for any region, the validity of the certificate remains inconclusive. The same procedure is applied to verify the initial set and unsafe set conditions.

We compare our method with that in \cite{zhang2023exact}, as it is, to the best of our knowledge, the only method that does not rely on the Heaviside step function assumption for the derivative of the ReLU activation. 
Notably, our method supports both verification and falsification of barrier certificates, whereas the method in \cite{zhang2023exact} only performs verification. Consequently, if a certificate fails their test, no conclusion can be drawn about its correctness. Furthermore, the correctness guarantees in \cite{zhang2023exact} hold only if all the involved optimization problems return optimal solutions. In practice, however, solvers often yield suboptimal results in nonlinear programming, potentially undermining the reliability of their verification.

\begin{figure}[bp]
\begin{subfigure}{.24\textwidth}
  \includegraphics[width= \textwidth]{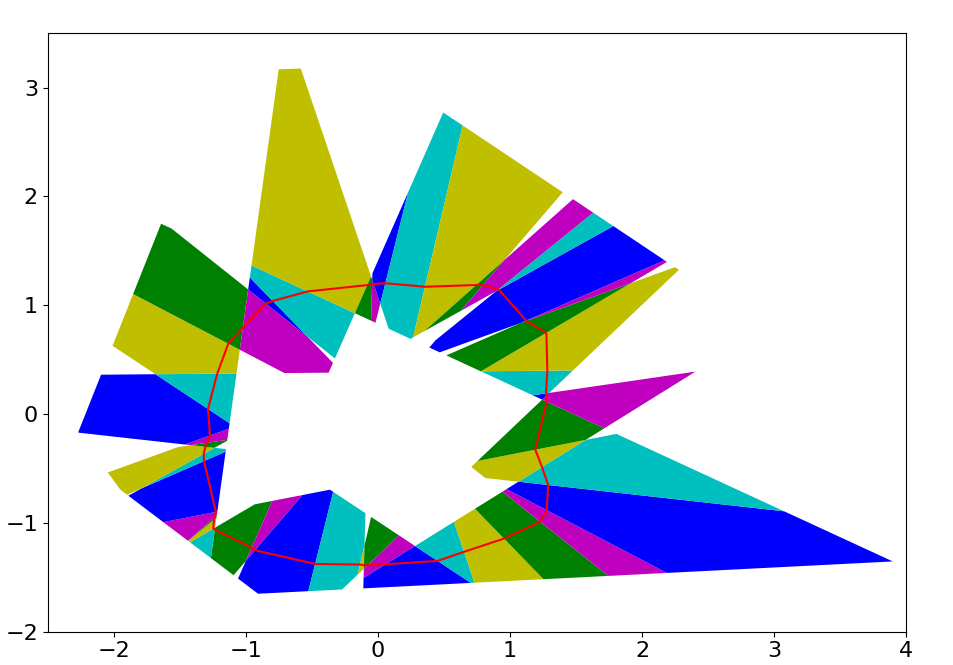}
  \subcaption*{(a)}
\end{subfigure}%
\begin{subfigure}{.24\textwidth}
  \includegraphics[width= \textwidth]{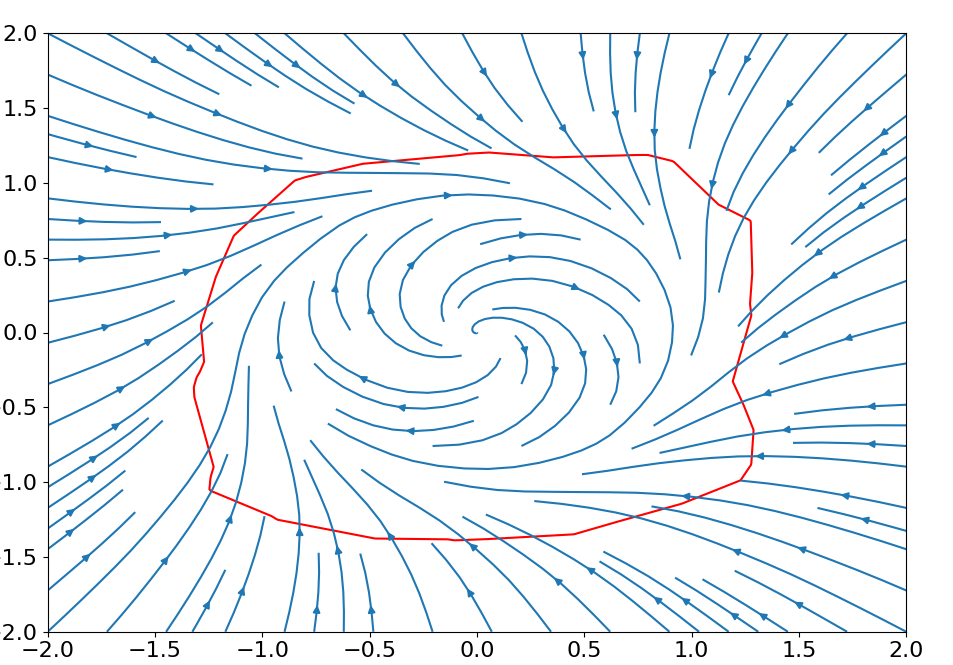}
  \subcaption*{(b)}
\end{subfigure}%
\caption{Enumerated valid linear regions
and the vector fields of  system  Arch3. 
}
\label{fig:Arch3Expo}
\end{figure}


The running time
and validity of barrier certificates for each case are summarized in Table \ref{tab:BFComparison}.  Our method successfully verifies or falsifies all cases, whereas the approach in \cite{zhang2023exact} can only verify five out of seven barrier certificates in the two-dimensional system Arch3. In other cases, it either returns ``unknown'' or crashes due to memory exhaustion. Furthermore, their results are reliable only when all optimal values of the underlying optimizations can be obtained, which is difficult to rigorously verify in nonlinear programming. 
 In terms of running time, 
our method consistently outperforms theirs, particularly in high-dimensional systems.
 This improvement stems from the fact that our algorithm enumerates far fewer regions than \cite{zhang2023exact}, as discussed in Remark \ref{remark: compare to hinge}. Figure \ref{fig:Arch3Expo} illustrates the valid linear regions enumerated by our boundary propagation algorithm for the first barrier certificate in system Arch3.

\begin{table}[!t]
\caption{
The first two columns list the system names and 
the  neural network architectures of the barrier certificates, where HNN denotes the number of neurons in each hidden layer. 
We denote the total running time and barrier certificate validity in our method as $t_{our}$ and $\text{V}_{our}$, respectively, and use $t_Z$ and $\text{V}_Z$ for the corresponding results from \cite{zhang2023exact}. A dash ($-$) indicates that the program crashed due to memory exhaustion. A~\cmark~indicates that the barrier certificate is verified (\textcolor{Maroon}{$\checkmark$} means verified under the assumption that all optimal values of the underlying optimizations can be obtained), while~\xmark~indicates falsification. A~\namark~denotes that the validity of the barrier certificate is unknown.
}

\label{tab:BFComparison}
\begin{tabular}{llllll}
\toprule
Case  & HNN  & $t_{our}$& $t_Z$ & $\text{V}_{our}$ & $\text{V}_{Z}$  \\ \cmidrule(r){1-6}
 \multirow{ 7}{*}{\thead{Arch3\\(2-d)\\(nonlinear)}}
 & 32   & 3.89 & 24.10 & \cmark  & \textcolor{Maroon}{$\checkmark$}      \\ 
 & 64   & 28.40 & 141.01 & \cmark & \textcolor{Maroon}{$\checkmark$}  \\ 
 & 96  & 94.75 & 528.36 & \cmark & \textcolor{Maroon}{$\checkmark$} \\
 & 128   & 204.78 & 198.17 & \xmark  & \namark \\
 & 32-32  & 19.75 & 41.10 & \cmark & \textcolor{Maroon}{$\checkmark$} \\
 & 64-64  & 215.99 & 283.00 & \cmark & \textcolor{Maroon}{$\checkmark$} \\
 & 96-96  & 774.27 & 1644.60 & \xmark & \namark\\
 \cmidrule(r){1-6}
     \multirow{ 7}{*}{\thead{Complex\\(3-d)\\(nonlinear)}} & 32  & 5.29 & 135.22 & \cmark & \namark     \\    
 & 64 & 92.36 & 698.35 & \cmark & \namark     \\     
 & 96  & 469.53 & 1377.00 & \cmark & \namark     \\
 & 32-32  & 79.57 & 1132.23 & \cmark & \namark     \\
 & 64-64  & 500.08 &$-$ & \xmark & $-$     \\
 & 96-96  & 19686.48 & $-$ & \cmark & $-$     \\
 \cmidrule(r){1-6}
 \multirow{7}{*}{\thead{Linear4d\\(4-d)\\(linear)}}
 & 16  &  9.61 & 286.46  & \cmark   &  \namark   \\
 & 32 &  88.12 & 1110.54  & \cmark  &  \namark   \\
 & 48  &  9992.90 & $-$ & \cmark   & $-$   \\
 & 8-8  &  7.84 & 212.19  & \cmark   &  \namark   \\
 & 12-12  &  55.05 & 593.78  & \cmark &  \namark   \\
 & 16-16  &  524.73 & 2243.71  & \xmark  &  \namark   \\
 & 24-24  &  22571.44 &$-$  & \xmark  &  $-$ \\

 \cmidrule(r){1-6}
 \multirow{4}{*}{\thead{Decay\\(6-d)\\(nonlinear)}}
 & 20  & 19842.22 &$-$  & \cmark   & $-$   \\
 & 8-8  & 347.72 & 4381.28  & \cmark   &  \namark   \\ 
 & 10-10  & 1953.08 & 4035.94  & \xmark  &  \namark   \\
 & 12-12  & 31112.56 & $-$  & \xmark  &  $-$   \\ 
 \bottomrule
\end{tabular}
\end{table}

\section{Conclusion}
\label{sec: conclusion}

In this paper, we proposed a necessary and sufficient condition for verifying ReLU neural barrier certificates and developed a corresponding algorithm that supports both verification and falsification. To the best of our knowledge, this is the first method capable of falsifying ReLU neural barrier certificates, thereby filling a critical gap for their reliable application in safety-critical systems. Through numerical experiments, we demonstrated that our method achieves superior verification performance compared to existing approaches, particularly in high-dimensional systems.

\bibliography{aaai2026}

\newpage
  \clearpage

\renewcommand\thesubsection{\Alph{subsection}}

\section*{Appendix}

\subsection{
Explanation for the Lack of Mathematical Rigor in the  Derivative Assumption}

In this section, we explain why assuming that the derivative of the ReLU function is the Heaviside function is mathematically non-rigorous for verifying positive invariance. To this end, we first characterize the tangent cones of the 0-superlevel set of a continuously differentiable function in the following lemma.
\begin{lemma}[{\cite[Theorem 2.4.7]{clarke1990optimization}}]
\label{lemma: single regular}
    Let $\mathcal{S}= \{\bm{z} \in \mathbb{R}^n: h(\bm{z}) \geq 0\}$, where  $h:\mathbb{R}^n \rightarrow \mathbb{R}$ is  continuously differentiable, and let $\bm{x}$ be a point satisfying $h(\bm{x}) =0$. Suppose that $\nabla h(\bm{x}) \neq \bm{0}$, then $\mathcal{S}$ is regular at $\bm{x}$ and 
    \begin{equation*}
        T_{\mathcal{S}}^B(\bm{x})=T_{\mathcal{S}}^C(\bm{x}) = \{\bm{v} \in \mathbb{R}^n: \nabla h(\bm{x}) \bm{v} \geq 0\}.
    \end{equation*}
\end{lemma}
Combining this lemma with Theorem \ref{Nagumo’s theorem}, we obtain the following proposition.
\begin{proposition}
\label{pro: diff h}
    Let $\mathcal{S}= \{\bm{z} \in \mathbb{R}^n: h(\bm{z}) \geq 0\}$,  where  $h:\mathbb{R}^n \rightarrow \mathbb{R}$ is  continuously differentiable. Suppose that for all $\bm{x}\in \mathcal{S}$ with $h(\bm{x}) =0$, we have $\nabla h(\bm{x}) \neq \bm{0}$. Then $\mathcal{S}$ is positively invariant under the system~\eqref{system} if and only if 
    \begin{equation*}
    \nabla h(\bm{x}) \bm{f}(\bm{x}) \geq 0, \forall \bm{x} \in \{\bm{x}\in \mathbb{R}^n:h(\bm{x}) =0\}.
    \end{equation*}  
\end{proposition}


Some existing work (e.g., \cite{zhao2022verifying}) applies Proposition \ref{pro: diff h} to verify the positive invariance of the 0-superlevel set of a ReLU neural barrier certificate. However, this is inconsistent with the requirement that $h$ is continuously differentiable, since a ReLU neural barrier certificate is inherently non-differentiable at certain points. Even if the derivative of ReLU function is forcibly defined as the Heaviside function, the latter’s discontinuity still violates the requirement of continuous differentiability. Therefore, Proposition \ref{pro: diff h}  is not applicable to the verification of ReLU neural barrier certificates.

\subsection{Background of Polyhedron}
\begin{definition}[Affine hull]
Given a set of points $\mathcal{S} = \{\bm{x}_1, \bm{x}_2, \ldots, \bm{x}_k\} \subset \mathbb{R}^n$, the affine hull of $\mathcal{S}$ is defined as
\begin{equation*}
    \operatorname{aff}(\mathcal{S}) = \left\{ \sum_{i=1}^k \lambda_i \bm{x}_i \,\middle|\, \sum_{i=1}^k \lambda_i = 1,\ \lambda_i \in \mathbb{R} \right\}.
\end{equation*}
The dimension of an affine hull is the dimension of the corresponding linear vector space.
\end{definition}
\begin{definition}[Polyhedron, \cite{ziegler2012lectures}]
   Polyhedron is an intersection of closed halfspaces: a set $\mathcal{P} \subseteq \mathbb{R}^d$ presented in the form
    $\mathcal{P}=\left\{\boldsymbol{x} \in \mathbb{R}^d: \bm{A}\bm{x} \leq\bm{d}\right\} \quad \text { for some } \bm{A} \in \mathbb{R}^{m \times d},\bm{d} \in \mathbb{R}^m$.
The dimension of a polyhedron is the dimension of its affine hull: $\operatorname{dim}( \mathcal{P})\triangleq\operatorname{dim}(\operatorname{aff}( \mathcal{P}))$.
\end{definition}

\begin{definition}[Face]
Let $ \mathcal{P} \subseteq \mathbb{R}^d$ be a convex polyhedron. A linear inequality $\bm{c}^\top \bm{x} \leq c_0$ is valid for $ \mathcal{P}$ if it is satisfied for all points $\boldsymbol{x} \in  \mathcal{P}$. A face of $ \mathcal{P}$ is any set of the form
\begin{equation*}
  \mathcal{F}= \mathcal{P} \cap\left\{\boldsymbol{x} \in \mathbb{R}^d: \bm{c}^\top \bm{x}=c_0\right\},   
\end{equation*}
where $\bm{c}^\top \bm{x} \leq c_0$ is a valid inequality for $ \mathcal{P}$. The dimension of a face is the dimension of its affine hull: $\operatorname{dim}( \mathcal{F})\triangleq\operatorname{dim}(\operatorname{aff}( \mathcal{F}))$.
\end{definition}

Faces of a polyhedron $\mathcal{P}$ with  dimensions 0, 1, $\dim( \mathcal{P})-2$, and $\dim( \mathcal{P})-1$ are called \textit{vertices}, \textit{edges}, \textit{ridges}, and \textit{facets}, respectively. The \textit{relative interior} of $\mathcal{P}$, denoted  $\relind(\mathcal{P})$, is defined as 
as the interior of $\mathcal{P}$ with respect to the embedding of $ \mathcal{P}$ into its affine hull $\aff(\mathcal{P})$, in which $\mathcal{P}$ is full-dimensional.

    

\subsection{Other Definitions of Tangent Cones}
The tangent cones in Definition \ref{def: tangent cone} can also be  defined in the sense of set-valued analysis: 

\begin{equation}
    T_\mathcal{S}^B(\bm{x})\triangleq\limsup_{t \rightarrow 0^{+}}  \frac{\mathcal{S}-\bm{x}}{t},    
\end{equation}

\begin{equation}
    T_\mathcal{S}^C(\bm{x})\triangleq\liminf _{t \rightarrow 0^{+}, \bm{x}^\prime \stackrel{\mathcal{S}}{\rightarrow} \bm{x}}  \frac{\mathcal{S}-\bm{x}}{t},
\end{equation}
where $\frac{\mathcal{S}-\bm{x}}{t} \triangleq
\{\frac{\bm{z}-\bm{x}}{t}, \bm{z} \in \mathcal{S}\} $, and for details on limits of sets and set-valued maps, we refer the reader to \cite{aubin2009set}.

Equivalently, the tangent cones can be characterized in terms of sequences:
\begin{itemize}
    \item $\bm{v} \in T_{\mathcal{S}}^B(\bm{x})$ if and only if $\exists t_n \rightarrow 0^+$ and $\exists \bm{v}_n \rightarrow \bm{v}$ such that $\forall n \in \mathbb{N}^+,\bm{x}+t_n\bm{v}_n \in \mathcal{S}$.

     \item $\bm{v} \in T_{\mathcal{S}}^C(\bm{x})$ if and only if $\forall t_n \rightarrow 0^+$, $\forall \bm{x}_n \stackrel{\mathcal{S}}{\rightarrow} \bm{x}$, $\exists \bm{v}_n \rightarrow \bm{v}$ such that $\forall n \in \mathbb{N}^+,\bm{x}_n+t_n\bm{v}_n \in \mathcal{S}$.
\end{itemize}

\subsection{Detailed Explanation for Remark \ref{remark: Bouligand cannot opt}}

In Proposition \ref{lemma: linear function b cone}, we state that for any $\bm{x} \in \partial \mathcal{C}$, 
\begin{itemize}
    \item[1.] if $\bm{x} \in \Int{\mathcal{X}_i},  i = 1,\cdots,  N$, then
    \begin{equation*}
    \begin{aligned}
    T_\mathcal{C}^B(\bm{x}) = 
            \{\bm{v}\in \mathbb{R}^n: \bm{w}_i^\top \bm{v} \geq 0\};
    \end{aligned}
        \end{equation*}
        \item[2.] if $\bm{x} \in \bigcap_{k=1}^{m} {\mathcal{X}_{i_k}} $ and $\bm{x} \not\in \overline{\mathbb{R}^n \setminus\bigcup_{k=1}^{m} \mathcal{X}_{i_k}}, i_1,\cdots, i_m \in \{1,2,\cdots, N\}$, then 
   \begin{equation*}
       \begin{aligned}
        &T_\mathcal{C}^B(\bm{x}) =\bigcup_{k=1}^m\{\bm{v}\in \mathbb{R}^n:
            \bigwedge_{j \in E} \bm{A}_{i_k}(j) \bm{v} \leq 0 \wedge \bm{w}_{i_k}^\top \bm{v} \geq 0\},
       \end{aligned}
   \end{equation*} 
        where the set $E$ is defined as $E \triangleq\{j \in \{1,\cdots, \rows(\bm{A})\}: \bm{A}_{i_k}(j) \bm{x} = \bm{d}_{i_k}(j)\}$.
\end{itemize}
Thus, verifying the condition ``for all $\bm{x} \in \partial \mathcal{S}$, $\bm{f}(\bm{x}) \in T_{\mathcal{S}}^B(\bm{x})$'' in Theorem~\ref{Nagumo’s theorem} can be divided into two parts: (1) verification within the interior of each linear region intersecting the boundary, and (2) verification at their intersections.

To verify the condition ``for all $\bm{x} \in \partial \mathcal{S} \cap \Int{\mathcal{X}_i}$, $\bm{f}(\bm{x}) \in T_{\mathcal{S}}^B(\bm{x})$'', the corresponding ``optimization formulation'' should be 
\begin{equation*}   
\begin{split}
     &\min_{\bm{x}} \bm{w}_i^\top \bm{f}(\bm{x})  \\
     \text{s.t.}~& \begin{cases}
              \bm{w}_i^\top \bm{x} + b(\mathscr{C}) =0,\\
        \bm{A}_i \bm{x} < \bm{d}_i.
     \end{cases}
    \end{split}
\end{equation*}
However, the strict inequalities ``$\bm{A}_i \bm{x} < \bm{d}_i$'' are not permitted  in standard optimization formulations (see \url{https://yalmip.github.io/inside/strictinequalities/}). To make the optimization  problem computationally tractable, these constraints  must be relaxed to  non-strict ones: $\bm{A}_i \bm{x} \leq \bm{d}_i$. 
This relaxation, however, transforms  the condition based on Bouligand tangent cone from necessary and sufficient to merely sufficient. As a result, the optimization method in \cite{zhang2023exact} can only verify, but not falsify, the positive invariance.

\subsection{Verifying and Falsifying $\mathcal{S}_I \cap \partial \mathcal{C} =\emptyset$ via SMT and Optimization}

We can translate \eqref{eq:verificaiton initial} into an SMT problem that determines whether the following quantifier-free formula with Disjunctive Normal Form (DNF) is satisfiable (SAT) or not (UNSAT).

\begin{equation}
\label{eq:verificaiton initial SMT} 
\begin{aligned}
        \bigvee_{\mathscr{C}\in \mathcal{A}} 
             &\left(\bm{w}(\mathscr{C})^\top \bm{x} + b(\mathscr{C}) =0 \wedge
         \bm{x} \in \mathcal{X}(\mathscr{C})  \wedge h_I(\bm{x}) > 0 \right).   
\end{aligned}
\end{equation}
If the SMT problem \eqref{eq:verificaiton initial SMT} is proven to be UNSAT, then the Boolean value of \eqref{eq:verificaiton initial} is true,  implying that  $\mathcal{S}_I \cap \partial \mathcal{C} =\emptyset$; In the opposite, when \eqref{eq:verificaiton initial SMT} is SAT, then  the Boolean value of \eqref{eq:verificaiton initial} is false,  implying that  $\mathcal{S}_I \cap \partial \mathcal{C} \neq \emptyset$.

The quantified formula \eqref{eq:verificaiton initial} can also be described as a series of optimization problems: 
for each  valid linear region $\mathcal{X}(\mathscr{C}), \mathscr{C} \in \mathcal{A}$, we solve 
\begin{equation}   
                 \label{eq: nonlinear verification for initial}
\begin{split}
     &\max_{\bm{x}} h_I(\bm{x})  \\
     \text{s.t.}~& \begin{cases}
              \bm{w}(\mathscr{C})^\top \bm{x} + b(\mathscr{C}) =0,\\
         \bm{x} \in \mathcal{X}(\mathscr{C}).
     \end{cases}
    \end{split}
\end{equation}
If all  the optimal values of \eqref{eq: nonlinear verification for initial} are non-positive for all $\mathscr{C} \in \mathcal{A}$, then $\mathcal{S}_I \cap \partial \mathcal{C} =\emptyset$.
If there exists a valid linear region $\mathcal{X}(\mathscr{C}), \mathscr{C} \in \mathcal{A}$  such that \eqref{eq: nonlinear verification for initial}  yields a positive feasible value,  then $\mathcal{S}_I \cap \partial \mathcal{C} \neq\emptyset$.

The verifying and falsifying of the condition $\mathcal{S}_U \cap \partial \mathcal{C} =\emptyset$ are similar to the ones for $\mathcal{S}_I \cap \partial \mathcal{C} =\emptyset$, just simply replacing $h_I$ with $h_U$ in \eqref{eq:verificaiton initial SMT} and \eqref{eq: nonlinear verification for initial}.

\subsection{Auxiliary Lemmas and Proofs}
\begin{lemma}
\label{lemma: open ball tangent cone}
    Let $\mathcal{S}$ be a closed subset of the Banach space $X$ with norm $\Vert \cdot \Vert$, for all $\epsilon >0$,
    \begin{equation}
    \begin{aligned}
        T_{\mathcal{S}}^B(\bm{x}) = T_{\mathcal{S} \cap B(\bm{x},\epsilon)}^B(\bm{x}).
    \end{aligned}
    \end{equation}
\end{lemma}

\begin{proof}
    $T_{\mathcal{S}}^B(\bm{x}) \supset T_{\mathcal{S} \cap B(\bm{x},\epsilon)}^B(\bm{x})$ (if $\mathcal{K} \subset \mathcal{L}$ and $\bm{x} \in \overline{\mathcal{K}}$, then $T_{\mathcal{K}}^B(\bm{x}) \subset T_{\mathcal{L}}^B(\bm{x})$, see \cite[Chapter 4]{aubin2009set}). Next we prove $T_{\mathcal{S}}^B(\bm{x}) \subset T_{\mathcal{S} \cap B(\bm{x},\epsilon)}^B(\bm{x})$.

    
    For $\bm{v} \in T_{\mathcal{S}}^B(\bm{x})$, there exists a sequence $\{t_n\}$ satisfying $t_n \rightarrow 0^+$, then $\exists N_1 >0$ such that $\forall n\geq N_1, t_n \leq \frac{\epsilon}{4\Vert \bm{v}\Vert}$, since $\bm{v} \in T_{\mathcal{S}}^B(\bm{x})$, we can find $\bm{v}_n \rightarrow \bm{v}$ such that $\forall n \in \mathbb{N}^+, \bm{x}+t_n\bm{v}_n \in \mathcal{S}$. For sequence $\{\bm{v}_n\}$, $\exists N_2 >0$ such that $\forall n \geq N_2, \Vert \bm{v}_n -\bm{v} \Vert \leq \Vert \bm{v} \Vert, \Vert \bm{v}_n \Vert \leq 2\Vert \bm{v} \Vert$. Let $N=\max\{N_1,N_2\}$, set a sequence
    \begin{equation*}
        \bm{v}^\prime_n = \begin{cases}
            \frac{t_N}{t_n} \bm{v}_N, 1 \leq n \leq N\\
            \bm{v}_n, n > N,
        \end{cases}
    \end{equation*} 
    then $\forall n \in \mathbb{N}^+, \Vert \bm{x}+t_n \bm{v}^\prime_n - \bm{x} \Vert \leq \frac{\epsilon}{2}$, thus $\forall n \in \mathbb{N}^+, \bm{x}+t_n \bm{v}^\prime_n \in \mathcal{S} \cap B(\bm{x},\epsilon)$, which means $\bm{v} \in T_{\mathcal{S} \cap B(\bm{x},\epsilon)}^B(\bm{x})$.
\end{proof}

\begin{lemma}[{\cite[Theorem 6.12 in Chapter 3]{clarke2008nonsmooth}}]
\label{lemma: relationship between tangent cones}
    Let $\mathcal{S}$ be a closed subset of the Banach space $X$, a vector $\bm{v}$ which lies in $T_\mathcal{S}^C(\bm{x})$ is one which lies in $T_\mathcal{S}^B\left(\bm{x}^{\prime}\right)$ for all $\bm{x}^{\prime}$ near $x$, i.e., $\bm{v} \in T_\mathcal{S}^C(\bm{x})$ iff
\begin{equation}
  \limsup _{\bm{x}^\prime \stackrel{\mathcal{S}}{\rightarrow} \bm{x}} d\left(\bm{v}, T_\mathcal{S}^B\left(\bm{x}^{\prime}\right)\right)=0,  
\end{equation}
where $\bm{x}^\prime \stackrel{\mathcal{S}}{\rightarrow} \bm{x}$ means the convergence is in $\mathcal{S}$.
\end{lemma}

\begin{lemma}
\label{lemma: tangent cone of single cone}
    For a polyhedral cone $\mathcal{P}=\{\bm{z} \in \mathbb{R}^n: \bm{A}\bm{z} \leq \bm{d}
    \}$, where the apex $\bm{x}$ satisfies $\bm{A}\bm{x} = \bm{d}$, the tangent cones to $\mathcal{P}$ at the apex $\bm{x}$ are 
    \begin{equation}
        T_\mathcal{P}^B(\bm{x}) =T_\mathcal{P}^C(\bm{x}) = \{ \bm{v} \in \mathbb{R}^n: \bm{A}\bm{v} \leq \bm{0}\}.
    \end{equation}
\end{lemma}

\begin{proof}
        $T_\mathcal{P}^B(\bm{x}) = \limsup_{t \rightarrow 0^{+}}  \frac{\mathcal{P}-\bm{x}}{t}$, set $ t\bm{v} =\bm{z}-\bm{x}$, then $\forall \bm{z}$ with $\bm{A}\bm{z} \leq \bm{d}$, $\bm{v}$ satisfies $\bm{A}(\bm{x}+t\bm{v})\leq \bm{d} \Leftrightarrow t\bm{A}\bm{v}\leq \bm{0}$. Then $\limsup_{t \rightarrow 0^{+}}  \frac{\mathcal{P}-\bm{x}}{t} = \limsup_{t \rightarrow 0^{+}} \{ \bm{v} \in \mathbb{R}^n: t\bm{A}\bm{v} \leq \bm{0}\}= \limsup_{t \rightarrow 0^{+}} \{ \bm{v} \in \mathbb{R}^n: \bm{A}\bm{v} \leq \bm{0}\} =\{ \bm{v} \in \mathbb{R}^n: \bm{A}\bm{v} \leq \bm{0}\}$. 

        Since $\mathcal{P}$ is a convex set, then $T_\mathcal{P}^C(\bm{x})=T_\mathcal{P}^B(\bm{x}) = \{ \bm{v} \in \mathbb{R}^n: \bm{A}\bm{v} \leq \bm{0}\}$ (according to \cite[Proposition 4.2.1]{aubin2009set}).
\end{proof}

\begin{lemma}
\label{lemma: tangent cone of polyhedral cone}
    For polyhedral cones $\mathcal{P}_1=\{\bm{z} \in \mathbb{R}^n: \bm{A}_1\bm{z} \leq \bm{d}_1
    \}, \cdots, \mathcal{P}_i=\{\bm{z} \in \mathbb{R}^n: \bm{A}_i\bm{z} \leq \bm{d}_{i}
    \}, \cdots, \mathcal{P}_n=\{\bm{z} \in \mathbb{R}^n: \bm{A}_n\bm{z} \leq \bm{d}_n
    \}$, they share a same apex $\bm{x}$ satisfying $\bm{A}_i\bm{z} = \bm{d}_{i}, \forall i =1,\cdots,n$, the Bouligand tangent cone to $\bigcup_{i=1}^n \mathcal{P}_i$ at $\bm{x}$ is \begin{equation}
        T_{\bigcup_{i=1}^n \mathcal{P}_i}^B(\bm{x})  = \bigcup_{i=1}^n\{ \bm{v} \in \mathbb{R}^n: \bm{A}_i\bm{v} \leq \bm{0}\}.
    \end{equation}
\end{lemma}

\begin{proof}
 Set $ t\bm{v} =\bm{z}-\bm{x}$, then  $\bm{z} = \bm{x}+t\bm{v}$, $\frac{\bigcup_{i=1}^n\mathcal{P}_i-\bm{x}}{t} = \bigcup_{i=1}^n \frac{\mathcal{P}_i-\bm{x}}{t} = \bigcup_{i=1}^n \{\bm{v} \in \mathbb{R}^n: \bm{A}_i(\bm{x}+t\bm{v}) \leq \bm{d}_i\} = \bigcup_{i=1}^n \{\bm{v} \in \mathbb{R}^n: t\bm{A}_i \bm{v} \leq \bm{0}\}$.  
 
 $T_{\bigcup_{i=1}^n \mathcal{P}_i}^B(\bm{x}) = \limsup_{t \rightarrow 0^{+}}  \frac{\bigcup_{i=1}^n\mathcal{P}_i-\bm{x}}{t} = \limsup_{t \rightarrow 0^{+}} \bigcup_{i=1}^n \{\bm{v} \in \mathbb{R}^n: t\bm{A}_i \bm{v} \leq \bm{0}\} = \bigcup_{i=1}^n \{\bm{v} \in \mathbb{R}^n: \bm{A}_i \bm{v} \leq \bm{0}\}$.
\end{proof}

\subsection{Proof of Proposition \ref{lemma: linear function b cone}}
\begin{proof}
 1. For $\bm{x} \in \Int{\mathcal{X}_i}$, the assertion can be proved by the following inclusion relationships.
 \begin{itemize}
     \item $\{ \bm{v}\in \mathbb{R}^n: \bm{w}_i^\top \bm{v} \geq 0\} \subset T_\mathcal{C}^B(\bm{x})$: $\forall \bm{v} \in \{ \bm{v}\in \mathbb{R}^n: \bm{w}_i^\top \bm{v} \geq 0\}$, $\exists \delta(\bm{v})>0$, $\forall t \leq \delta(\bm{v}), \bm{x} + t\bm{v} \in \mathcal{X}_i$, $h(\bm{x} + t\bm{v})=\bm{w}_i(\bm{x} + t\bm{v}) + b_i = t\bm{w}_i^\top \bm{v} \geq 0$, which means $\bm{x} + t\bm{v} \in \mathcal{C}$, $d(\bm{x} + t\bm{v}, \mathcal{C}) = 0$, therefore $\liminf _{t \rightarrow 0^{+}}\frac{d(\bm{x}+t \bm{v}, \mathcal{S}) }{t}=0$, $\bm{v} \in T_\mathcal{C}^B(\bm{x})$.
     \item $\{ \bm{v}\in \mathbb{R}^n: \bm{w}_i^\top \bm{v} \geq 0\} \supset T_\mathcal{C}^B(\bm{x})$: suppose that there exists $\bm{v}^\prime$, $\bm{v}^\prime \in T_\mathcal{C}^B(\bm{x})$ and $\bm{v}^\prime \not\in \{ \bm{v}\in \mathbb{R}^n: \bm{w}_i^\top \bm{v} \geq 0\}$, $\exists \delta>0$, $\forall t \leq \delta$, $\bm{x} + t\bm{v}^\prime\in \mathcal{X}_i$, then $h(\bm{x} + t\bm{v}^\prime)=\bm{w}_i^\top(\bm{x} + t\bm{v}^\prime) + b_i = t\bm{w}_i^\top \bm{v}^\prime < 0$, which means $\bm{x} + t\bm{v}^\prime \not\in \mathcal{C}$,  $\forall t \leq \frac{\epsilon}{2}, d(\bm{x} + t\bm{v}^\prime, \mathcal{C}) = \frac{\vert \bm{w}_i^\top(\bm{x} + t\bm{v}^\prime) + b_i\vert}{\Vert \bm{w}_i\Vert} =\frac{-\bm{w}_i^\top t\bm{v}^\prime}{\Vert \bm{w}_i\Vert}$, $\liminf _{t \rightarrow 0^{+}}\frac{d(\bm{x}+t \bm{v}^\prime, \mathcal{S}) }{t}=\frac{-\bm{w}_i^\top \bm{v}^\prime}{\Vert \bm{w}_i\Vert}>0$, which is contrary to the assumption $\bm{v}^\prime \in T_\mathcal{C}^B(\bm{x})$, thus $\{ \bm{v}\in \mathbb{R}^n: \bm{w}_i^\top \bm{v} \geq 0\} \supset T_\mathcal{C}^B(\bm{x})$.
\end{itemize}
2. For $\bm{x}$ on the intersection of linear regions, i.e., $\bm{x} \in \bigcap_{k=1}^{m} {\mathcal{X}_{i_k}} \wedge \bm{x} \not\in \overline{\mathbb{R}^n \setminus\bigcup_{k=1}^{m} \mathcal{X}_{i_k}}, i_1,\cdots, i_m \in \{1,2,\cdots, N\}$, since $\bm{x} \not\in \overline{\mathbb{R}^n \setminus\bigcup_{k=1}^{m} \mathcal{X}_{i_k}}$, then $\bm{x} \in \Int \bigcup_{k=1}^{m} \mathcal{X}_{i_k}$, there exists $\epsilon>0$ such that $B(\bm{x},\epsilon) \subset \bigcup_{k=1}^{m} \mathcal{X}_{i_k}$. 

 
 We define 
 $\Tilde{\mathcal{C}} \triangleq \mathcal{C} \cap B(\bm{x},\epsilon)$, 
 $\tilde{\mathcal{X}}_{i_k}\triangleq \{\bm{z} \in \mathbb{R}^n: \bigwedge_{j \in E} \bm{A}_{i_k}(j) \bm{z} \leq \bm{d}_{i_k}(j)\}$, where $E =\{j \in \{1,\cdots, \rows(A)\}: \bm{A}_{i_k}(j) \bm{x} = \bm{d}_{i_k}(j)\}$, for $k \in\{1,\cdots,m\}$, in other word, $\tilde{\mathcal{X}}_{i_k}$ remains the linear constraints passing through $\bm{x}$. In this way, all $\tilde{\mathcal{X}}_{i_k}$ are polyhedral cones sharing the same apex $\bm{x}$ and $\bigcup_{k=1}^m \tilde{\mathcal{X}}_{i_k}= \mathbb{R}^n$.

Let $\tilde{\mathcal{X}}_{i_k}^+ \triangleq \{\bm{z}\in \mathbb{R}^n:
            \bigwedge_{j \in E} \bm{A}_{i_k}(j) \bm{z} \leq \bm{d}_{i_k} \wedge \bm{w}_{i_k}^\top \bm{z} + b_{i_k} \geq 0\}$, since $\bm{x} \in \bigcap_{k=1}^{m} {\mathcal{X}_{i_k}}$, $\forall k\in\{1,\cdots,m\}, \bm{w}_{i_k}^\top \bm{x} + b_{i_k} = 0$, then all $\tilde{\mathcal{X}}_{i_k}^+$ are polyhedral cones with apex $\bm{x}$, further, $\tilde{\mathcal{C}} = \bigcup_{k=1}^{m} {\tilde{\mathcal{X}}_{i_k}}^+ \cap B(\bm{x}, \epsilon)$.   
According to Lemma \ref{lemma: open ball tangent cone} and Lemma \ref{lemma: tangent cone of polyhedral cone}, we have
    \begin{equation*}
    \begin{aligned}
         T_{\mathcal{C}}^B(\bm{x})=&T_{\mathcal{C} \cap B(\bm{x},\epsilon)}^B(\bm{x})=T_{\bigcup_{k=1}^{m} {\tilde{\mathcal{X}}_{i_k}}^+ \cap B(\bm{x},\epsilon)}^B(\bm{x}) \\=&T_{\bigcup_{k=1}^{m} {\tilde{\mathcal{X}}_{i_k}}^+}^B(\bm{x}) = \bigcup_{k=1}^{m}T_{\mathcal{X}_{i_k}^+}^B(\bm{x})\\=&\bigcup_{k=1}^m\{\bm{v}\in \mathbb{R}^n:
            \bigwedge_{j \in E} \bm{A}_{i_k}(j) \bm{v} \leq 0 \wedge \bm{w}_{i_k}^\top \bm{v} \geq 0\},
    \end{aligned}
    \end{equation*}

 Hereinafter we prove $T_\mathcal{C}^B(\bm{x}) \supset  \{\bm{v}\in \mathbb{R}^n: \bigwedge_{k\in I}
            \bm{w}_{i_k}^\top \bm{v} \geq 0\}$.
            
 For a region $\tilde{\mathcal{X}}_{i_k}$, it has three forms of relationships with the hyperplane $\mathcal{H}_{i_k}=\{\bm{z} \in \mathbb{R}^n: \bm{w}_{i_k}^\top \bm{z} + b_{i_k}=0\}$:
\begin{itemize}
    \item $\mathcal{H}_{i_k}$ has the intersection with the interior of $\tilde{\mathcal{X}}_{i_k}$, i.e., $\mathcal{H}_{i_k} \cap \Int \tilde{\mathcal{X}}_{i_k} \neq \emptyset$.
    \item $\mathcal{H}_{i_k} \cap \tilde{\mathcal{X}}_{i_k}$ is a facet of $\tilde{\mathcal{X}}_{i_k}$.
    \item $\mathcal{H}_{i_k} \cap \tilde{\mathcal{X}}_{i_k}$ is a face but not facet of $\tilde{\mathcal{X}}_{i_k}$.
\end{itemize}

For the first form, the hyperplane $\mathcal{H}_{i_k}$ can divided $\tilde{\mathcal{X}}_{i_k}$ into two parts, $\tilde{\mathcal{X}}_{i_k} \cap \{\bm{z} \in \mathbb{R}^n: \bm{w}_{i_k}^\top \bm{z} + b_{i_k}\geq 0\}$ and $\tilde{\mathcal{X}}_{i_k} \cap \{\bm{z} \in \mathbb{R}^n: \bm{w}_{i_k}^\top \bm{z} + b_{i_k}\leq 0\}$, which are marked as $\tilde{\mathcal{X}}_{i_k}^+$ and $\tilde{\mathcal{X}}_{i_k}^-$ respectively. Notice that $\forall \bm{z} \in \tilde{\mathcal{X}}_{i_k}^+, h(\bm{z}) \geq 0$ and $\forall \bm{z} \in \tilde{\mathcal{X}}_{i_k}^-, h(\bm{z}) \leq 0$.

For the second form, there are two situations $\tilde{\mathcal{X}}_{i_k} \subset  \{\bm{z} \in \mathbb{R}^n: \bm{w}_{i_k}^\top \bm{z} + b_{i_k}\geq 0\}$ and $\tilde{\mathcal{X}}_{i_k} \subset   \{\bm{z} \in \mathbb{R}^n: \bm{w}_{i_k}^\top \bm{z} + b_{i_k}\leq 0\}$, for the first situation, $\forall \bm{z} \in \tilde{\mathcal{X}}_{i_k}, h(\bm{z}) \geq 0$, and $\forall \bm{z} \in \tilde{\mathcal{X}}_{i_k}, h(\bm{z}) \leq 0$ for the second.

For the third form, there are also two situations same as the second form.

We define the index sets as follows:
\begin{itemize}
    \item $I_{1} \triangleq \{k \in \{1,\cdots,m\}$: $\tilde{\mathcal{X}}_{i_k}$ belongs to the first form$\}$;
    \item $I_{21} \triangleq \{k \in \{1,\cdots,m\}: \tilde{\mathcal{X}}_{i_k}$ belongs to the second form and $\tilde{\mathcal{X}}_{i_k} \subset  \{\bm{z} \in \mathbb{R}^n: \bm{w}_{i_k}^\top \bm{z} + b_{i_k}\geq 0\}\}$;
   \item $I_{22} \triangleq \{k \in \{1,\cdots,m\}: \tilde{\mathcal{X}}_{i_k}$ belongs to the second form and $\tilde{\mathcal{X}}_{i_k} \subset  \{\bm{z} \in \mathbb{R}^n: \bm{w}_{i_k}^\top \bm{z} + b_{i_k}\leq 0\}\}$;
    \item $I_{31} \triangleq \{k \in \{1,\cdots,m\}: \tilde{\mathcal{X}}_{i_k}$ belongs to the third form and $\tilde{\mathcal{X}}_{i_k} \subset  \{\bm{z} \in \mathbb{R}^n: \bm{w}_{i_k}^\top \bm{z} + b_{i_k}\geq 0\}\}$;
   \item $I_{32} \triangleq \{k \in \{1,\cdots,m\}: \tilde{\mathcal{X}}_{i_k}$ belongs to the third form and $\tilde{\mathcal{X}}_{i_k} \subset  \{\bm{z} \in \mathbb{R}^n: \bm{w}_{i_k}^\top \bm{z} + b_{i_k}\leq 0\}\}$;
\end{itemize}
We denote that 
\begin{equation*}
    \begin{aligned}
        \tilde{\mathcal{X}}^+ \triangleq \bigcup_{k\in I_1} \tilde{\mathcal{X}}_{i_k}^+ \cup \bigcup_{k\in I_{21} \cup I_{31}} \tilde{\mathcal{X}}_{i_k},
        \\\tilde{\mathcal{X}}^- \triangleq \bigcup_{k\in I_1} \tilde{\mathcal{X}}_{i_k}^- \cup \bigcup_{k\in I_{22} \cup I_{32}} \tilde{\mathcal{X}}_{i_k}.
    \end{aligned}
\end{equation*}
Notice that $\tilde{\mathcal{X}}^+ \cup \tilde{\mathcal{X}}^- = \mathbb{R}^n$, $\partial \tilde{\mathcal{X}}^+ = \partial \tilde{\mathcal{X}}^-= \tilde{\mathcal{X}}^+ \cap \tilde{\mathcal{X}}^-$. Moreover, for each $l \in I_1$, $\tilde{\mathcal{X}}^+_{i_l}$ and $\tilde{\mathcal{X}}^-_{i_l}$ are adjacent and share same facets; $|I_{21}| = |I_{22}|$, for each linear regions $\tilde{\mathcal{X}}_{i_p}$ where $p \in I_{21}$, there exists $q \in I_{22}$ such that $\tilde{\mathcal{X}}_{i_p}$ and $\tilde{\mathcal{X}}_{i_q}$ are adjacent and share same facets. All the facets discussed above form the boundary of $\tilde{\mathcal{X}}^+$ and $\tilde{\mathcal{X}}^-$, i.e., $\partial \tilde{\mathcal{X}}^+ = \partial \tilde{\mathcal{X}}^-=\bigcup_{k\in I}(\{\bm{z} \in \mathbb{R}^n: \bm{w}_{i_k}^\top \bm{z} + b_{i_k} = 0\} \cap \tilde{\mathcal{X}}_{i_k})$, where $I =\{k \in \{1,\cdots,m\}: \mathcal{X}_{i_k} \text{is a valid linear region}\}$, in other word, $I=I_1 \cup I_{21} \cup I_{22}$.

Next we prove $\tilde{\mathcal{C}} = \tilde{\mathcal{X}}^+ \cap B(\bm{x}, \epsilon)$.
\begin{itemize}
    \item $\tilde{\mathcal{C}} \supset \tilde{\mathcal{X}}^+ \cap B(\bm{x}, \epsilon)$: $\forall \bm{z} \in \tilde{\mathcal{X}}^+ \cap B(\bm{x},\epsilon), h(\bm{z}) \geq 0$, then $\tilde{\mathcal{C}} \supset \tilde{\mathcal{X}}^+ \cap B(\bm{x}, \epsilon)$.
    
    \item $\tilde{\mathcal{C}} \subset \tilde{\mathcal{X}}^+ \cap B(\bm{x}, \epsilon)$: we have $\forall \bm{z} \in \tilde{\mathcal{X}}^-\cap B(\bm{x},\epsilon), h(\bm{z}) \leq 0$, $\tilde{\mathcal{X}}^+ \cup \tilde{\mathcal{X}}^- =\mathbb{R}^n$. 
    
    For all $\bm{z}$ satisfying $\bm{z} \in B(\bm{x}, \epsilon)$ and $h(\bm{z}) > 0$, since $\bm{z}^\prime \not\in \tilde{\mathcal{X}}^- \cap B(\bm{x}, \epsilon)$, then $\bm{z} \in \tilde{\mathcal{X}}^+ \cap B(\bm{x}, \epsilon)$.
    
    Suppose that there exists $\bm{z}^\prime$ such that $\bm{z}^\prime \in   B(\bm{x}, \epsilon) \setminus \tilde{\mathcal{X}}^+$ and $h(\bm{z}^\prime) = 0$, since $\Int \mathcal{C} \cup \partial \mathcal{C} = \overline{\mathcal{C}} = \mathcal{C} = \overline{\Int \mathcal{C}} = \Int \mathcal{C} \cup \partial(\Int \mathcal{C}), \Int \mathcal{C} \cap \partial \mathcal{C} = \emptyset, \Int \mathcal{C} \cap \partial(\Int \mathcal{C}) = \emptyset$,  then $\partial \mathcal{C} = \partial(\Int \mathcal{C})$. $\bm{z}^\prime \in \partial \mathcal{C} = \partial(\Int \mathcal{C})$, choose a positive number $\delta$ such that $B(\bm{z}^\prime, \delta) \subset  B(\bm{x}, \epsilon) \setminus \tilde{\mathcal{X}}^+$, we have $B(\bm{z}^\prime, \delta) \cap \Int{\mathcal{C}} \neq \emptyset$, thus there exists a point $\bm{z}^{\prime\prime}$ such that $h(\bm{z}^{\prime\prime}) >0 $ and $\bm{z}^{\prime\prime} \in B(\bm{x}, \epsilon) \setminus \tilde{\mathcal{X}}^+ \subset \tilde{\mathcal{X}}^- \cap B(\bm{x}, \epsilon)$, however, this is contrary to the fact that $\forall \bm{z} \in \tilde{\mathcal{X}}^-\cap B(\bm{x},\epsilon), h(\bm{z}) \leq 0$. Therefore for all $\bm{z}$ satisfying $\bm{z} \in B(\bm{x}, \epsilon)$ and $h(\bm{z}) = 0$, $\bm{z} \in \tilde{\mathcal{X}}^+ \cap B(\bm{x}, \epsilon)$.

    To sum up, for all $\bm{z}$ satisfying $\bm{z} \in B(\bm{x}, \epsilon)$ and $h(\bm{z}) \geq 0$, $\bm{z} \in \tilde{\mathcal{X}}^+ \cap B(\bm{x}, \epsilon)$, i.e., $\tilde{\mathcal{C}} \subset \tilde{\mathcal{X}}^+ \cap B(\bm{x}, \epsilon)$.
\end{itemize}

    Since $\tilde{\mathcal{C}} = \tilde{\mathcal{X}}^+ \cap B(\bm{x}, \epsilon)$, in $B(\bm{x}, \epsilon)$, $\partial \mathcal{C} = \partial \tilde{\mathcal{X}}^+  = \bigcup_{k\in I}
            (\{\bm{z} \in \mathbb{R}^n: \bm{w}_{i_k}^\top \bm{z} + b_{i_k} = 0\} \cap \tilde{\mathcal{X}}_{i_k})$, where $I =\{k \in \{1,\cdots,m\}: \mathcal{X}_{i_k} \text{ is a valid linear region}\}$.
            With the fact that $\tilde{\mathcal{X}}^+$ is the union of cones sharing the same apex $\bm{x}$, then the boundary of $\tilde{\mathcal{X}}^+$ in $\mathbb{R}^n$ is also $\partial \tilde{\mathcal{X}}^+ = \bigcup_{k\in I}
            (\{\bm{z} \in \mathbb{R}^n: \bm{w}_{i_k}^\top \bm{z} + b_{i_k} = 0\} \cap \tilde{\mathcal{X}}_{i_k})$. 
            
            We define a new cone with the apex $\bm{x}$, $\mathcal{P} \triangleq \{\bm{x} +\bm{z}: \bigwedge_{k\in I}
            \bm{w}_{i_k}^\top \bm{z} \geq 0\}$. 
            The constraints for $\mathcal{P}$ can be written as a conjunctive clause $\bigwedge_{j=1}^{|I|} c_j$, where all the literals $c_j$ have the form $\bm{w}^\top \bm{z} \geq 0$, then $\mathcal{P} = \{\bm{x} +\bm{z}: \bigwedge_{j=1}^{|I|} c_j\}$. For $\tilde{\mathcal{X}}^+$, it is the union of cones, i.e., $\tilde{\mathcal{X}}^+ = \bigcup_{l=1} \mathcal{P}_l$, every $\mathcal{P}_l$ can be written as $\mathcal{P} = \{\bm{x} +\bm{z}: \bigwedge_{j=1}^{M_l} c_j^{(l)}\}$, where $c_j^{(l)}$ also has the form $\bm{w}^\top \bm{z} \geq 0$, ${M_l}$ is the number of constraints of $\mathcal{P}_l$. We assume here that every 
            constraint $c_j^{(l)}$ of each $\mathcal{P}_l$ is not a redundant constraint, i.e.,  $\forall \bm{z}((\bigwedge_{k=1, k\neq j}^{M_l} c_k^{(l)} ) \rightarrow c_j^{(l)})$ is false.
          $\tilde{\mathcal{X}}^+$ can be written as $\tilde{\mathcal{X}}^+ = \{ \bm{x}+ \bm{z}: \bigvee_{l=1}^{|I_1|+|I_{21}|+|I_{31}|} \bigwedge_{j=1}^{M_l} c_j^{(l)}\}$ and we have  
        \begin{equation*}
            \mathcal{P} \cap \tilde{\mathcal{X}}^+ = \{\bm{x}+\bm{z}:  (\bigwedge_{j=1}^{|I|} c_j) \wedge (\bigvee_{l=1}^{|I_1|+|I_{21}|+|I_{31}|} \bigwedge_{j=1}^{M_l} c_j^{(l)})\}.
        \end{equation*}
    For every $l$, in clause $A_l: \bigwedge_{j=1}^{M_l} c_j^{(l)}$, if it has the literals in $\bigwedge_{j=1}^{|I|} c_j$, removing those literals wouldn't change  the propositional formula $(\bigwedge_{j=1}^{|I|} c_j) \wedge (\bigvee_{l=1}^{|I_1|+|I_{21}|+|I_{31}|} \bigwedge_{j=1}^{M_l} c_j^{(l)})$, i.e., 
    \begin{equation*}
        (\bigwedge_{j=1}^{|I|} c_j) \wedge (\bigvee_{l=1}^{|I_1|+|I_{21}|+|I_{31}|} A_l) \Leftrightarrow (\bigwedge_{j=1}^{|I|} c_j) \wedge (\bigvee_{l=1}^{|I_1|+|I_{21}|+|I_{31}|} A^\prime_l),
    \end{equation*}
    where $A^\prime_l$ is a clause that removes literals in $\bigwedge_{j=1}^{|I|} c_j$ from $A_l$, for example, if $\bigwedge_{j=1}^{|I|} c_j: c_1 \wedge c_2 \wedge c_3$, $A_l: c_1 \wedge c^\prime \wedge c^{\prime\prime}$, then $A^\prime_l: c^\prime \wedge c^{\prime\prime}$. 

\begin{equation*}
    \begin{aligned}
        \mathcal{P} \cap \tilde{\mathcal{X}}^+ =& \{\bm{x}+\bm{z}:  (\bigwedge_{j=1}^{|I|} c_j) \wedge (\bigvee_{l=1}^{|I_1|+|I_{21}|+|I_{31}|} A_l)\} \\ =&\{\bm{x}+\bm{z}:(\bigwedge_{j=1}^{|I|} c_j) \wedge (\bigvee_{l=1}^{|I_1|+|I_{21}|+|I_{31}|} A^\prime_l)\}\\
=& \mathcal{P} \cap \{\bm{x}+\bm{z}: (\bigvee_{l=1}^{|I_1|+|I_{21}|+|I_{31}|} A^\prime_l)\}.
    \end{aligned}
\end{equation*}
Let $\mathcal{G} \triangleq \{\bm{x}+\bm{z}: (\bigvee_{l=1}^{|I_1|+|I_{21}|+|I_{31}|} A^\prime_l)\}$, recall that $\partial \tilde{\mathcal{X}}^+ = \bigcup_{k\in I}
            (\{\bm{z} \in \mathbb{R}^n: \bm{w}_{i_k}^\top \bm{z} + b_{i_k} = 0\} \cap \tilde{\mathcal{X}}_{i_k})$, then $\mathcal{G}$ removes all boundary constraints of $\tilde{\mathcal{X}}^+$, which means $\mathcal{G}= \mathbb{R}^n$, then $\mathcal{P} \cap \tilde{\mathcal{X}}^+ = \mathcal{P} \cap \mathcal{G} =\mathcal{P} \cap \mathbb{R}^n = \mathcal{P}$, $\mathcal{P} \subset \tilde{\mathcal{X}}^+$. According to Lemma \ref{lemma: open ball tangent cone} and properties of Bouligand tangent cone (if $\mathcal{K} \subset \mathcal{L}$ and $\bm{x} \in \overline{\mathcal{K}}$, then $T_{\mathcal{K}}^B(\bm{x}) \subset T_{\mathcal{L}}^B(\bm{x})$, see \cite[Chapter 4]{aubin2009set}), we have

 \begin{equation*}
     \begin{aligned}
         T_{\mathcal{P}}^B(\bm{x}) \subset T_{\tilde{\mathcal{X}}^+}^B(\bm{x})&= T_{\tilde{\mathcal{X}}^+ \cap B(\bm{x},\epsilon)}^B(\bm{x}) \\&=T_{\mathcal{C} \cap B(\bm{x},\epsilon)}^B(\bm{x}) = T_{\mathcal{C}}^B(\bm{x}).
     \end{aligned}
 \end{equation*}     

 By Lemma \ref{lemma: tangent cone of polyhedral cone}, $T_{\mathcal{P}}^B(\bm{x})= \{\bm{v}\in \mathbb{R}^n: \bigwedge_{k\in I}
            \bm{w}_{i_k}^\top \bm{v} \geq 0\}$, therefore $T_\mathcal{C}^B(\bm{x}) \supset  \{\bm{v}\in \mathbb{R}^n: \bigwedge_{k\in I}
            \bm{w}_{i_k}^\top \bm{v} \geq 0\}$.

3. In the proof of the second assertion, for $\bm{x} \in \bigcap_{k=1}^{m} {\mathcal{X}_{i_k}} \wedge \bm{x} \not\in \overline{\mathbb{R}^n \setminus\bigcup_{k=1}^{m} \mathcal{X}_{i_k}}, i_1,\cdots, i_m \in \{1,2,\cdots, N\}$, there exists $\epsilon>0$ such that $B(\bm{x},\epsilon) \subset \bigcup_{k=1}^{m} \mathcal{X}_{i_k}$, we have proven that in $B(\bm{x}, \epsilon)$, $\partial \mathcal{C} = \partial \tilde{\mathcal{X}}^+  = \bigcup_{k\in I}
            (\{\bm{z} \in \mathbb{R}^n: \bm{w}_{i_k}^\top \bm{z} + b_{i_k} = 0\} \cap \tilde{\mathcal{X}}_{i_k})$, therefore, $\partial \mathcal{C} \cap B(\bm{x}, \epsilon) \subset \bigcup_{k\in I}\mathcal{X}_{i_k}$. In fact, when $m=1$, the conclusion still holds, i.e., for $\bm{x} \in \Int \mathcal{X}_j, j \in \{1,\cdots,N\}$, there exists $\epsilon>0$ such that $B(\bm{x},\epsilon) \subset \mathcal{X}_j$, $\partial \mathcal{C} \cap B(\bm{x}, \epsilon) \subset \mathcal{X}_j$, here $\mathcal{X}_j$ must be a valid linear region since $\Int \mathcal{X}_j \cap \partial \mathcal{C} \neq \emptyset$.
            
 Here we give an uniform description, let $K(\bm{x}) =  \{l \in \{1,\cdots,N\}: \bm{x} \in \mathcal{X}_l\}$, $J =\{l \in \{1,\cdots,N\}: \mathcal{X}_{l} \text{ is a valid linear region}\}$, for each $\bm{x} \in \partial \mathcal{C}$,
$\bm{x} \in\bigcap_{l \in K(\bm{x})} {\mathcal{X}_{l}} \wedge \bm{x} \not\in \overline{\mathbb{R}^n \setminus\bigcup_{l \in K(\bm{x})} \mathcal{X}_{l}}$, there exists $\epsilon(\bm{x})>0$ such that $B(\bm{x},\epsilon(\bm{x})) \subset \bigcup_{l \in K(\bm{x})} \mathcal{X}_{l}$, $\partial \mathcal{C} \cap B(\bm{x}, \epsilon(\bm{x})) \subset \bigcup_{l\in K(\bm{x}) \cap J}\mathcal{X}_{l}$. Therefore, $\partial \mathcal{C} = \bigcup_{\bm{x} \in \partial\mathcal{C}} (\partial \mathcal{C} \cap B(\bm{x}, \epsilon(\bm{x}))) \subset \bigcup_{\bm{x} \in \partial \mathcal{C}} \bigcup_{l \in K(\bm{x}) \cap J} \mathcal{X}_l = \bigcup_{l \in J} \mathcal{X}_l$.
            
\end{proof}

\subsection{Proof of Proposition \ref{lemma: linear function condition}}
\begin{proof}
 For the first case, according to Lemma \ref{lemma: relationship between tangent cones}, $\bm{v} \in T_\mathcal{C}^C(\bm{x})$ iff
$$
\limsup _{\bm{x}^\prime \stackrel{\mathcal{C}}{\rightarrow} \bm{x}} d\left(\bm{v}, T_\mathcal{C}^B\left(\bm{x}^{\prime}\right)\right)=0,
$$

Since $\bm{x} \in \Int{\mathcal{X}_i}$, there exists $\delta >0$ such that $\forall \epsilon \leq \delta$, $B(\bm{x}, \epsilon) \subset \Int \mathcal{X}_i,$ then $\forall \epsilon \leq \delta, \forall \bm{x}^\prime \in B(\bm{x}, \epsilon), T_\mathcal{C}^B(\bm{x}^\prime)= T_\mathcal{C}^B(\bm{x}) = \{ \bm{v}\in \mathbb{R}^n: \bm{w}_i^\top \bm{v} \geq 0\}.$ Then,
\begin{equation*}
    \begin{aligned}
      &\limsup _{\bm{x}^\prime \stackrel{\mathcal{C}}{\rightarrow} \bm{x}} d(\bm{v}, T_\mathcal{C}^B(\bm{x}^\prime)) \\
      & = \lim_{\epsilon \rightarrow 0^+}(\sup \{ d(\bm{v}, T_\mathcal{C}^B(\bm{x}^\prime)):\bm{x}^\prime \in (B(\bm{x},\epsilon) \cap \mathcal{C})\setminus \{\bm{x}\}\})  \\
      &=\lim_{\epsilon \rightarrow 0^+}(\sup \{ d(\bm{v}, T_\mathcal{C}^B(\bm{x}^\prime)):\bm{x}^\prime \in (B(\bm{x},\epsilon) \cap \partial \mathcal{C})\setminus \{\bm{x}\}\}) \\
      &= \sup \{ d(\bm{v}, T_\mathcal{C}^B(\bm{x}^\prime)):\bm{x}^\prime \in (B(\bm{x},\delta) \cap \partial \mathcal{C})\setminus \{\bm{x}\}\} \\
      & = d(\bm{v}, T_\mathcal{C}^B(\bm{x})).
    \end{aligned}
\end{equation*}
With the fact that $T_\mathcal{C}^B(\bm{x})$ is closed, $\bm{v} \in T_\mathcal{C}^B(\bm{x}) \Leftrightarrow d(\bm{v}, T_\mathcal{C}^B(\bm{x})) =0$, then $\bm{v} \in T_\mathcal{C}^C(\bm{x})\Leftrightarrow \bm{v} \in T_\mathcal{C}^B(\bm{x}),   T_\mathcal{C}^C(\bm{x}) = T_\mathcal{C}^B(\bm{x}) = \{ \bm{v}\in \mathbb{R}^n: \bm{w}_i^\top \bm{v} \geq 0\}.$

 For the second case, $\bm{x}$ on the intersection of linear regions, i.e., $\bm{x} \in \bigcap_{k=1}^{m} {\mathcal{X}_{i_k}} \wedge \bm{x} \not\in \overline{\mathbb{R}^n \setminus\bigcup_{k=1}^{m} \mathcal{X}_{i_k}}, i_1,\cdots, i_m \in \{1,2,\cdots, N\}$, since $\bm{x} \not\in \overline{\mathbb{R}^n \setminus\bigcup_{k=1}^{m} \mathcal{X}_{i_k}}$, then $\bm{x} \in \Int \bigcup_{k=1}^{m} \mathcal{X}_{i_k}$, there exists $\delta>0$ such that $B(\bm{x},\delta) \subset \bigcup_{k=1}^{m} \mathcal{X}_{i_k}$.

Herein we use the notations $\tilde{\mathcal{X}}_{i_k}, I_{1}, I_{21}, I_{22}$ defined in the proof of Proposition \ref{lemma: linear function b cone}.

Given any closed ball $B(\bm{x}, \epsilon)$ with $\epsilon \leq \delta$, For each $\tilde{\mathcal{X}}_{i_j}$ with $j \in I_{1}$, there always exists $\bm{x}^\prime \in B(\bm{x}, \epsilon) \cap \partial \mathcal{C} \cap \Int \tilde{\mathcal{X}}_{i_j}$.

Given any closed ball $B(\bm{x}, \epsilon)$ with $\epsilon \leq \delta$, For each $\tilde{\mathcal{X}}_{i_p}$ with $p \in I_{21}$ (or $p \in I_{22}$), there always exists $\bm{x}^\prime \in B(\bm{x}, \epsilon) \cap \relind(\mathcal{P}_{i_p})$, where $\mathcal{P}_{i_p} = \partial \mathcal{C} \cap \tilde{\mathcal{X}}_{i_p} = \{\bm{z}\in \mathbb{R}^n: \bm{w}_{i_p}^\top \bm{z} + b_{i_p} =0\} \cap \tilde{\mathcal{X}}_{i_p}$ is a facet of $\tilde{\mathcal{X}}_{i_p}$, further, there exists another linear region $\tilde{\mathcal{X}}_{i_q}$ with $q \in I_{22}$ (or $q \in I_{21}$) such that $\tilde{\mathcal{X}}_{i_p}$ and $\tilde{\mathcal{X}}_{i_q}$ share the same facet $\partial \mathcal{C} \cap \tilde{\mathcal{X}}_{i_p} = \partial \mathcal{C} \cap \tilde{\mathcal{X}}_{i_q}$. Here $\bm{x}^\prime \in \mathcal{X}_{i_p} \cap \mathcal{X}_{i_q} \wedge \bm{x}^\prime \not\in \overline{\mathbb{R}^n \setminus (\mathcal{X}_{i_p} \cap \mathcal{X}_{i_q})}$, according to Proposition \ref{lemma: linear function b cone}, $T_\mathcal{C}^B(\bm{x}^\prime) = \{ \bm{v} \in \mathbb{R}^n: \bm{w}_{i_p}^\top \bm{v} \geq 0\} = \{ \bm{v} \in \mathbb{R}^n: \bm{w}_{i_q}^\top \bm{v} \geq 0\}$.

\begin{equation*}
    \begin{aligned}
      &\limsup _{\bm{x}^\prime \stackrel{\mathcal{C}}{\rightarrow} \bm{x}} d(\bm{v}, T_\mathcal{C}^B(\bm{x}^\prime)) \\
      &=\lim_{\epsilon \rightarrow 0^+}(\sup \{ d(\bm{v}, T_\mathcal{C}^B(\bm{x}^\prime)):\bm{x}^\prime \in (B(\bm{x},\epsilon) \cap \mathcal{C})\setminus \{\bm{x}\}\}) \\
      &=\lim_{\epsilon \rightarrow 0^+}(\sup \{ d(\bm{v}, T_\mathcal{C}^B(\bm{x}^\prime)):\bm{x}^\prime \in (B(\bm{x},\epsilon) \cap \partial \mathcal{C})\setminus \{\bm{x}\}\}) 
    \end{aligned}
\end{equation*}
When $\epsilon \leq \delta$, 
\begin{equation*}
\begin{aligned}
     &\lim_{\epsilon \rightarrow 0^+}(\sup \{ d(\bm{v}, T_\mathcal{C}^B(\bm{x}^\prime)):\bm{x}^\prime \in (B(\bm{x},\epsilon) \cap \partial \mathcal{C})\setminus \{\bm{x}\}\}) =\\
     &\sup \{ d(\bm{v}, T_\mathcal{C}^B(\bm{x}^\prime)):\bm{x}^\prime \in (B(\bm{x},\delta) \cap \partial \mathcal{C})\setminus \{\bm{x}\}\} =\\
    &\max\left\{   
    \begin{aligned}
        &\sup \left\{ d(\bm{v}, T_\mathcal{C}^B(\bm{x}^\prime)):\bm{x}^\prime \in \Int \mathcal{X}_{i_j}\right\}, j \in I_1; \\
    &\sup \left\{ d(\bm{v}, T_\mathcal{C}^B(\bm{x}^\prime)):\bm{x}^\prime \in \mathcal{X}_{i_p} \cap \mathcal{X}_{i_q}  \right.\\
    &\left. \wedge \bm{x}^\prime \not\in\overline{\mathbb{R}^n \setminus (\mathcal{X}_{i_p} \cap \mathcal{X}_{i_q})}\right\}, p \in I_{21}, q \in I_{22};\\
    &\sup \left\{ d(\bm{v}, T_\mathcal{C}^B(\bm{x}^\prime)):\bm{x}^\prime  \in \bigcap_{l \in L} {\mathcal{X}_{i_l}} \right.\\
    &\left. \wedge \bm{x}^\prime \not\in \overline{\mathbb{R}^n \setminus\bigcup_{l\in L} \mathcal{X}_{i_l}}\right\},L \in 2^{\{1,\cdots,m\}}, L \neq \emptyset.
    \end{aligned}\right.
\end{aligned}
\end{equation*}
For simplicity, we omit the condition $\bm{x}^\prime \in (B(\bm{x},\delta) \cap \partial \mathcal{C})\setminus \{\bm{x}\}$ in the three cases of the `max' expression.  $\bm{v} \in T_\mathcal{C}^C(\bm{x})$ iff $\bm{v}$ lies in all Bouligand tangent cones appeared in the `max' expression, i.e., 
\begin{equation*}
\begin{aligned}
 \bm{v} \in& \bigcap_{j \in I_1}\{\bm{v}: \bm{w}_{i_j}^\top \bm{v} \geq 0\} \cap \bigcap_{p \in I_{21}}\{\bm{v}: \bm{w}_{i_p}^\top \bm{v} \geq 0\} 
  \\&\cap\bigcap_{q \in I_{22}}\{\bm{v}: \bm{w}_{i_q}^\top \bm{v} \geq 0\}  \\&\cap \bigcap_{\bm{x}^\prime  \in \bigcap_{l \in L} {\mathcal{X}_{i_l}} \wedge \bm{x}^\prime \not\in \overline{\mathbb{R}^n \setminus\bigcup_{l\in L} \mathcal{X}_{i_l}}} T_\mathcal{C}^B(\bm{x}^\prime)
\end{aligned}
\end{equation*}
where $L \in 2^{\{1,\cdots,m\}}, L \neq \emptyset$. According to \eqref{subeq: insection b cone subset} in Proposition \ref{lemma: linear function b cone}, we have for all $\bm{x}^\prime  \in \bigcap_{l \in L} {\mathcal{X}_{i_l}} \wedge \bm{x}^\prime \not\in \overline{\mathbb{R}^n \setminus\bigcup_{l\in L} \mathcal{X}_{i_l}}$, $\bigcap_{k \in I} \{\bm{v}: \bm{w}_{i_k}^\top \bm{v} \geq 0\} \subset T_\mathcal{C}^B(\bm{x}^\prime)$, where $I = I_1 \cup I_{21} \cup I_{22}$. Therefore, $\bm{v} \in T_\mathcal{C}^C(\bm{x}) \Leftrightarrow \bm{v} \in \{\bm{v}\in \mathbb{R}^n: \bigwedge_{k\in I}
            \bm{w}_{i_k}^\top \bm{v} \geq 0\}$.
\end{proof}

\subsection{Proof of Theorem \ref{thm: linear function invariance}}
\begin{proof}
    For all $\bm{x} \in \partial \mathcal{C}$, if there exists $i \in \{1,2, \cdots, N\}$ such that $\bm{x} \in \Int{\mathcal{X}_i}$, obviously $\mathcal{X}_i$ is a valid linear region, according to \eqref{eq: linear tangent cone}, $\bm{f}(\bm{x}) \in T_\mathcal{C}^C(\bm{x})$; if there exists $i_1,\cdots, i_m \in \{1,2,\cdots, N\}$ such that $\bm{x} \in \bigcap_{k=1}^{m} {\mathcal{X}_{i_k}} \wedge \bm{x} \not\in \overline{\mathbb{R}^n \setminus\bigcup_{k=1}^{m} \mathcal{X}_{i_k}}$, according to \eqref{eq: linear invariance condition}, $\bm{f}(\bm{x})$ satisfies $\bigwedge_{k\in I}
            \bm{w}_{i_k}^\top \bm{v} \geq 0$,  where $I =\{k \in \{1,\cdots,m\}: \mathcal{X}_{i_k} \text{is a valid linear region}\}$, in combination with \eqref{eq: linear tangent cone}, we have $\bm{f}(\bm{x}) \in T_\mathcal{C}^C(\bm{x})$. By the equivalence of assertions (a) and (c) in Theorem \ref{Nagumo’s theorem}, the conclusion holds.
\end{proof}

\subsection{Proof of Theorem \ref{thm: bound connected}}
\begin{proof}
    Let the valid linear region indicator set $\mathcal{A} = \{ \mathscr{C}_1, \cdots, \mathscr{C}_n\}, n \geq 2$, suppose there exists $\mathscr{C}^\prime \in \mathcal{A}$ such that it can't be found by the boundary propagation algorithm, then for each facet $\mathcal{F}$ of $\mathcal{X}(\mathscr{C}^\prime)$ such that $\mathcal{F} \cap \partial \mathcal{C}\neq \emptyset$, $\mathcal{F} \cap \partial \mathcal{C}$ doesn't belong to any other valid linear region, therefore the regions which have intersection with $\mathcal{X}(\mathscr{C}^\prime)\cap\partial\mathcal{C}$ are not valid linear regions, hence 
    $\partial\mathcal{C}$ can be partitioned into two nonempty open subsets: one containing $\partial\mathcal{C} \cap \mathcal{X}(\mathscr{C}^\prime)$ and the other containing $\bigcup_{\mathscr{C}\in \mathcal{A} \setminus \{\mathscr{C}^\prime\}}\partial\mathcal{C} \cap \mathcal{X}(\mathscr{C})$. However, this is contrary to the assumption that $\partial \mathcal{C}$ is a connected set.
\end{proof}

\subsection{Proof of Proposition \ref{prop: initial condition}}

\begin{proof}
    For the ``$\Leftarrow$'' direction: Suppose that $\mathcal{S}_I \not\subset \mathcal{C}$, then $\mathcal{S}_I \cap \mathcal{C}^c \neq \emptyset$. Moreover we have $\mathcal{S}_I \cap \Int \mathcal{C} \neq \emptyset$, Besides $(\mathcal{S}_I \cap \Int\mathcal{C})\cup (\mathcal{S}_I \cap \mathcal{C}^c) = \mathcal{S}_I \cap (\Int\mathcal{C} \cup \partial \mathcal{C} \cup\mathcal{C}^c) = \mathcal{S}_I$. However the above is contrary to the fact that $\mathcal{S}_I$ is connected (a connected set is a set that cannot be partitioned into two nonempty subsets which are open in the relative topology induced on the set).

    For the ``$\Rightarrow$'' direction: Suppose that there exist a point $\bm{x}$ such that $\bm{x} \in \mathcal{S}_I \cap \partial \mathcal{C}$, since $h_I(\bm{x}) >0$ and $h_I$ is continuous, there exists $\delta >0$ such that for all $\bm{z} \in B(\bm{x},\delta)$, $h_I(\bm{z}) >0$, i.e., $B(\bm{x},\delta) \subset \mathcal{S}_I$, besides $\mathcal{S}_I \subset \mathcal{C}$, then we have $B(\bm{x},\delta) \subset \mathcal{C}$, however this is contrary to $\bm{x} \in \partial \mathcal{C}$, therefore $\mathcal{S}_I \cap \partial \mathcal{C} = \emptyset$. Since $\mathcal{S}_I \neq \emptyset, \mathcal{S}_I \subset \mathcal{C}$ and $\mathcal{S}_I \cap \partial \mathcal{C} = \emptyset$, therefore $\mathcal{S}_I \cap \Int \mathcal{C} \neq \emptyset$.
\end{proof}

\subsection{Proof of Proposition \ref{prop: unsafe condition}}

\begin{proof}
    For the ``$\Leftarrow$'' direction: Suppose that  $\mathcal{S}_U \cap \mathcal{C} \neq \emptyset$, since $\mathcal{S}_U \cap \partial \mathcal{C} =\emptyset$, then $\mathcal{S}_U \cap \Int\mathcal{C} \neq \emptyset$. Moreover we have $\mathcal{S}_U \cap  \mathcal{C}^c \neq \emptyset$, Besides $(\mathcal{S}_U \cap \Int\mathcal{C})\cup (\mathcal{S}_U \cap \mathcal{C}^c) = \mathcal{S}_U \cap (\Int\mathcal{C} \cup \partial \mathcal{C} \cup\mathcal{C}^c) = \mathcal{S}_U$. However the above is contrary to the fact that $\mathcal{S}_U$ is connected (a connected set is a set that cannot be partitioned into two nonempty subsets which are open in the relative topology induced on the set).

    For the ``$\Rightarrow$'' direction: Suppose that there exist a point $\bm{x}$ such that $\bm{x} \in \mathcal{S}_U \cap \partial \mathcal{C}$, since $h_U(\bm{x}) >0$ and $h_U$ is continuous, there exists $\delta >0$ such that for all $\bm{z} \in B(\bm{x},\delta)$, $h_U(\bm{z}) >0$, i.e., $B(\bm{x},\delta) \subset \mathcal{S}_U$, besides $\mathcal{S}_U \subset \mathcal{C}^c$, then we have $B(\bm{x},\delta) \subset \mathcal{C}^c$, however this is contrary to $\bm{x} \in \partial \mathcal{C}$, therefore $\mathcal{S}_U \cap \partial \mathcal{C} = \emptyset$. Since $\mathcal{S}_U \neq \emptyset$ and $\mathcal{S}_U \subset \mathcal{C}^c$ , therefore $\mathcal{S}_U \cap \mathcal{C}^c \neq \emptyset$.
\end{proof}

\subsection{System Descriptions}
\label{appendix:Experiment_Details}
In this section, we provide descriptions of the systems employed in our experiments.

\textbf{Arch3:} The system Arch3 is a two-dimensional nonlinear polynomial system that is commonly used as a benchmark for barrier certificate synthesis in the literature.
The dynamics of  Arch3  are given by:

\begin{equation*}
    \left[\begin{array}{c}
    \dot{x}_1 \\
    \dot{x}_2
    \end{array}\right]=\left[\begin{array}{c}
    x_1-x_1^3+x_2-x_1x_2^2 \\
    -x_1+x_2-x_1^2x_2-x_2^3
    \end{array}\right].
\end{equation*}

\textbf{Complex:} We propose a newly designed system, denoted as  Complex, which is a three-dimensional nonlinear non-polynomial system involving both trigonometric and exponential terms. 
The dynamics of  Complex are given by:
\begin{equation*}
    \left[\begin{array}{c}
    \dot{x}_1 \\
    \dot{x}_2 \\
    \dot{x}_3
    \end{array}\right]=\left[\begin{array}{c}
    -x_1(1+\sin^2(x_2)+e^{-x_3^2}) \\
    -x_2(1+\cos^2(x_3)+\tanh(x_1^2) \\
    -x_3(1+\ln(1+x_1^2+x_2^2))
    
    \end{array}\right].
\end{equation*}


\textbf{Linear4d:} This is a four-dimensional linear system. Due to its linear nature, LPs can be used for both verification and falsification.
The dynamics of  Linear4d are given by:
\begin{equation*}
    \left[\begin{array}{c}
    \dot{x}_1 \\
    \dot{x}_2 \\
    \dot{x}_3 \\
    \dot{x}_4
    \end{array}\right]=\left[\begin{array}{c}
    -x_1 \\
    x_1-2x_2 \\
    x_1-4x_3 \\
    x_1-3x_4
    
    \end{array}\right].
\end{equation*}



\textbf{Decay:} This is a six-dimensional system exhibiting decay in each component. It serves to 
demonstrate the scalability of our method to higher-order systems.
The dynamics of  Decay are given by:

\begin{equation*}
    \left[\begin{array}{c}
    \dot{x}_1 \\
    \dot{x}_2 \\
    \dot{x}_3 \\
    \dot{x}_4 \\
    \dot{x}_5 \\
    \dot{x}_6 
    \end{array}\right]=\left[\begin{array}{c}
    -x_1(1 + x_1^2 + x_2^2 + x_3^2+x_4^2+x_5^2+x_6^2) \\
    -x_2(1 + x_1^2 + x_2^2 + x_3^2+x_4^2+x_5^2+x_6^2) \\
    -x_3(1 + x_1^2 + x_2^2 + x_3^2+x_4^2+x_5^2+x_6^2)  \\
    -x_4(1 + x_1^2 + x_2^2 + x_3^2+x_4^2+x_5^2+x_6^2)  \\
    -x_5(1 + x_1^2 + x_2^2 + x_3^2+x_4^2+x_5^2+x_6^2)  \\
    -x_6(1 + x_1^2 + x_2^2 + x_3^2+x_4^2+x_5^2+x_6^2) 
    
    \end{array}\right].
\end{equation*}

\subsection{Experimental Settings}
The neural barrier certificates used for verification are trained following the method in \cite{zhao2020synthesizing}, with the corresponding code available on GitHub\footnote{\url{https://github.com/zhaohj2017/HSCC20-Repeatability}}. 


For comparison with \cite{zhang2023exact}, we set the domain of each system as $[-3,3]^n$, the initial set as 
$\{\bm{x} =(x_i)_{1\leq i \leq n}\in \mathbb{R}^n: -\sum_{i=1}^n x_i^2 +0.04 \geq0\}$ and the unsafe set as $\{\bm{x} =(x_i)_{1\leq i \leq n}\in \mathbb{R}^n: -\sum_{i=1}^n (x_i-3)^2 +1 \geq0\}$ for all the systems, where $n$ is the dimension of system. We remark that our method can also be applied to unbounded cases.





The verification of barrier certificates was conducted on a Hyper-V virtual machine running Ubuntu 22.04 LTS (Jammy Jellyfish), configured with dynamic RAM allocation with at least 512 MB and 12 virtual CPUs. The host system was a 64-bit Windows PC equipped with a 12th Generation Intel(R) Core(TM) i9-12900K processor (3.20 GHz), 64 GB of RAM. 

All optimization problems in this work are formulated using the Python package Pydrake \cite{drake} on Ubuntu 22.04 LTS (Jammy Jellyfish) and solved with the nonlinear optimization solver NLopt \cite{NLopt}.


\subsection{Detailed Experiment Results}
\label{appendix:experiment_results}
We present additional experimental results that were omitted from the main text due to space constraints. Table~\ref{tab:BFTest} reports both the enumeration time and the number of valid linear regions. We observe that the enumeration time increases with the number of valid linear regions. 
The verification time for the positive invariance condition, reported in Table~\ref{tab:BFVerification}, is computed as the sum of the time spent solving the optimization and SMT problems, both of which are also reported individually. In all cases, the SMT solving time is consistently greater than the optimization time, with a particularly large gap observed for the system Decay.
Table~\ref{tab:BFInitialUnsafe} presents the satisfaction of the initial and unsafe conditions, together with the time required to verify them. Note that the results are reported only for barrier certificates that satisfy the positive invariance condition. All such certificates are also observed to satisfy the initial and unsafe conditions. Across the three tables, the cost of enumeration clearly dominates the overall computation time, compared to the time required for verification.

\begin{table}[htbp]
\caption{This table  reports the neural network architecture of barrier certificates, the number of valid linear regions ($N$), and the enumeration time ($t_e$). Here, $\sigma$ denotes the ReLU activation function.}

\label{tab:BFTest}
\centering
\begin{tabular}{llll}
\toprule
Case & NN Architecture  & $N$ & $t_e$    \\ \cmidrule(r){1-4}
 \multirow{ 7}{*}{Arch3}
 &  2-32-$\sigma$-1         & 48 & 3.64     \\
 &  2-64-$\sigma$-1  & 98 & 27.53     \\
 &  2-96-$\sigma$-1         & 146 & 92.87      \\
 &  2-128-$\sigma$-1        & 188 & 204.62      \\  
 &  2-32-$\sigma$-32-$\sigma$-1 & 80 & 18.93       \\ 
 &  2-64-$\sigma$-64-$\sigma$-1 & 206 & 214.68        \\ 
 & 2-96-$\sigma$-96-$\sigma$-1 & 348 & 772.49        \\  
 &  2-128-$\sigma$-128-$\sigma$-1 & 408 & 1650.81         \\ 
 \cmidrule(r){1-4}
 \multirow{ 7}{*}{Complex} & 3-32-$\sigma$-1 & 168 & 3.95      \\ 
 & 3-64-$\sigma$-1 & 1109 & 77.28        \\ 
 & 3-96-$\sigma$-1 & 2584 & 415.72       \\ 
 & 3-32-$\sigma$-32-$\sigma$-1 & 878 & 66.85       \\ 
 & 3-64-$\sigma$-64-$\sigma$-1 & 2069 & 498.18   \\
& 3-96-$\sigma$-96-$\sigma$-1 & 15558 & 18506.32    \\ 
 \cmidrule(r){1-4}
 \multirow{7}{*}{Linear4d}
 & 4-16-$\sigma$-1  &  308 & 8.29      \\
 & 4-32-$\sigma$-1 &  1204 & 80.18     \\
 & 4-48-$\sigma$-1  & 8762 &  9912.56   \\
 & 4-8-$\sigma$-8-$\sigma$-1  &  246 & 5.85       \\
 & 4-12-$\sigma$-12-$\sigma$-1  &  880 & 48.82     \\
 & 4-16-$\sigma$-16-$\sigma$-1  &  2699 & 522.99      \\
 & 4-24-$\sigma$-24-$\sigma$-1  &  11522 & 22561.19   \\
 \cmidrule(r){1-4}
 \multirow{4}{*}{Decay}
  & 6-20-$\sigma$-1 & 7087 & 18601.99    \\
 & 6-8-$\sigma$-8-$\sigma$-1 & 687 & 46.40     \\ 
 & 6-10-$\sigma$-10-$\sigma$-1 & 3052 & 1950.76    \\ 
 & 6-12-$\sigma$-12-$\sigma$-1 & 7728 & 31099.56       \\ 
 \bottomrule
\end{tabular}
\end{table}

\begin{table}[H]
\caption{
This table reports the time for solving optimization problems ($t_{opt}$), SMT problems ($t_{dReal}$), and their total time ($t_v$). The symbol “$-$” indicates that the SMT procedure is skipped: for the linear system Arch3, since the optimization alone can both verify and falsify, there is no need to invoke SMT; for other nonlinear systems, when the optimization already falsifies positive invariance, SMT processing is also unnecessary.
}
\label{tab:BFVerification}
\centering
\begin{tabular}{lllll}
\toprule
Case & NN Architecture  & $t_{opt}$ & $t_{dReal}$ & $t_v$   \\ \cmidrule(r){1-5}
 \multirow{ 7}{*}{Arch3}
 &  2-32-$\sigma$-1         & 0.02 & 0.07 & 0.09    \\
 &  2-64-$\sigma$-1 & 0.05 & 0.26 & 0.31   \\
 &  2-96-$\sigma$-1        & 0.08 & 0.58 & 0.67     \\
 &  2-128-$\sigma$-1        & 0.12 & $-$ & 0.12    \\  
 &  2-32-$\sigma$-32-$\sigma$-1 & 0.05 & 0.26 & 0.31     \\ 
 &  2-64-$\sigma$-64-$\sigma$-1& 0.15& 1.10& 1.26      \\ 
 & 2-96-$\sigma$-96-$\sigma$-1 & 0.61 & $-$ & 0.61       \\  
  \cmidrule(r){1-5}
 \multirow{ 6}{*}{Complex} & 3-32-$\sigma$-1 & 0.08 & 0.43 & 0.51      \\ 
 & 3-64-$\sigma$-1 & 0.66 & 5.10 & 5.76      \\ 
 & 3-96-$\sigma$-1 & 1.79 & 18.35 & 20.15      \\ 
 & 3-32-$\sigma$-32-$\sigma$-1 & 0.64 & 4.19 & 4.83       \\ 
 & 3-64-$\sigma$-64-$\sigma$-1 & 1.85 & $-$ & 1.85 \\
& 3-96-$\sigma$-96-$\sigma$-1 & 32.55 & 436.13 & 318.38  \\ 
 \cmidrule(r){1-5}
 \multirow{7}{*}{Linear4d}
 & 4-16-$\sigma$-1  & 0.11 & $-$  & 0.11    \\
 & 4-32-$\sigma$-1 & 0.49  & $-$  & 0.49     \\
 & 4-48-$\sigma$-1  & 5.87 & $-$ & 5.87    \\
 & 4-8-$\sigma$-8-$\sigma$-1  & 0.11 & $-$  & 0.11    \\
 & 4-12-$\sigma$-12-$\sigma$-1  & 0.45 & $-$  & 0.45   \\
 & 4-16-$\sigma$-16-$\sigma$-1  & 1.73  &$-$ & 1.73    \\
 & 4-24-$\sigma$-24-$\sigma$-1  & 10.18 & $-$  & 10.18 \\
 \cmidrule(r){1-5}
 \multirow{4}{*}{Decay}
  & 6-20-$\sigma$-1 & 4.33 & 724.62  & 729.99   \\
 & 6-8-$\sigma$-8-$\sigma$-1 & 0.47 & 183.49 & 183.96     \\ 
 & 6-10-$\sigma$-10-$\sigma$-1 & 2.31 & $-$ & 2.31    \\ 
 & 6-12-$\sigma$-12-$\sigma$-1 & 12.99 & $-$ & 12.99     \\ 
 \bottomrule
\end{tabular}
\end{table}

\begin{table}[H]
\caption{This table  reports the validity tags of the initial and unsafe conditions, denoted as $V_I$ and $V_U$, respectively, along with the total verification time, denoted as $t$. A ~\cmark~ indicates  that the barrier certificate satisfies the corresponding condition.}

\label{tab:BFInitialUnsafe}
\centering
\begin{tabular}{lllll}
\toprule
Case & NN Architecture  & $V_I$ & $V_U$ & $t$   \\ \cmidrule(r){1-5}
 \multirow{ 5}{*}{Arch3}
 &  2-32-$\sigma$-1         & \cmark & \cmark & 0.15    \\
 &  2-64-$\sigma$-1  & \cmark & \cmark & 0.54   \\
 &  2-96-$\sigma$-1         & \cmark & \cmark & 1.19     \\
 &  2-32-$\sigma$-32-$\sigma$-1 & \cmark & \cmark & 0.49     \\ 
 &  2-64-$\sigma$-64-$\sigma$-1 & \cmark & \cmark & 1.28      \\  \cmidrule(r){1-5}
 \multirow{ 5}{*}{Complex} & 3-32-$\sigma$-1 & \cmark & \cmark & 0.82      \\ 
 & 3-64-$\sigma$-1 & \cmark & \cmark & 9.31      \\ 
 & 3-96-$\sigma$-1 & \cmark & \cmark & 33.63      \\ 
 & 3-32-$\sigma$-32-$\sigma$-1 & \cmark & \cmark & 7.87       \\ 
& 3-96-$\sigma$-96-$\sigma$-1 & \cmark & \cmark & 742.12  \\ 
 \cmidrule(r){1-5}
 \multirow{5}{*}{Linear4d}
 & 4-16-$\sigma$-1 & \cmark & \cmark  & 1.21    \\
 & 4-32-$\sigma$-1 & \cmark & \cmark  & 7.44     \\
 & 4-48-$\sigma$-1  & \cmark & \cmark & 74.45    \\
 & 4-8-$\sigma$-8-$\sigma$-1 & \cmark & \cmark  & 1.87   \\
 & 4-12-$\sigma$-12-$\sigma$-1  & \cmark & \cmark  & 5.77   \\
 \cmidrule(r){1-5}
 \multirow{2}{*}{Decay}
  & 6-20-$\sigma$-1 & \cmark & \cmark & 510.24   \\
 & 6-8-$\sigma$-8-$\sigma$-1 & \cmark & \cmark & 117.21     \\ 
 \bottomrule
\end{tabular}
\end{table}

\end{document}